\documentclass[final]{dmtcs-episciences}
\usepackage[ruled,linesnumbered,lined,boxed,procnumbered]{algorithm2e}
\usepackage{tikz,float,comment,amsfonts,complexity,graphicx,soul}
\usepackage{amsmath,amssymb}
\usepackage{algorithmic}
\usepackage{float}
\usepackage{thm-restate}
\usepackage{geometry} 
\usepackage[round]{natbib}
\usetikzlibrary{decorations.pathreplacing}

\declaretheorem[name=Observation]{observation}
\declaretheorem[name=Lemma]{lemma}
\declaretheorem[name=Definition]{definition}
\newtheorem{theorem}{Theorem}
\definecolor{lightergray}{gray}{1}
\SetCommentSty{mycommfont}
\newcommand{\ehig}{\textsf{EHIG}}
\begin{document}
\author{K.K. Nisha\affiliationmark{1} \and N.S. Narayanaswamy\affiliationmark{1} \and S.M. Dhannya\affiliationmark{2}\thanks{Part of the work was done as a PhD student at Indian Institute of Technology Madras, Chennai, India}}
\affiliation{Indian Institute of Technology Madras, Chennai, India.\\
Sri Sivasubramaniya Nadar College of Engineering, Kalavakkam, Chennai, India}
\title{Exactly Hittable Interval Graphs}


\publicationdata{vol. 25:3 special issue ICGT'22}{2023}{6}{10.46298/dmtcs.10762}{2023-01-03; 2023-01-03; 2023-09-12}{2023-10-31}

\maketitle

\begin{abstract}

Given a set system (also well-known as a hypergraph) $H = \{\mathcal{U},\mathcal{X}\}$, where $\mathcal{U}$ is a set of elements and $\mathcal{X}$ is a set of subsets of $\mathcal{U}$, an exact hitting set $S$ is a subset of $\mathcal{U}$ such that each subset in $\mathcal{X}$ contains exactly one element in $S$.
 We refer to a set system as \textit{exactly hittable} if it has an exact hitting set. In this paper, we study interval graphs which have intersection models that are exactly hittable. We refer to these interval graphs as \textit{Exactly Hittable Interval Graphs} ($\ehig$). We present a forbidden structure characterization for $\ehig$. We also show that the class of proper interval graphs is a strict subclass of $\ehig$. Finally, we give an algorithm that runs in polynomial time to recognize graphs belonging to the class of $\ehig$.
\end{abstract}

\section{Introduction}\label{sec:GenFramework}

We study classes of simple graphs which are intersection graphs of set systems that have exact hitting sets. In particular, we introduce a class of interval graphs which are intersection graphs of intervals that have exact hitting sets. We refer to this class as 
\textit{Exactly Hittable Interval Graphs} (\textsf{EHIG}). 
We also present an infinite family of forbidden structures for \textsf{EHIG}. In the following, we introduce a setting of exact hitting sets and intersection graphs, before presenting our results. \\

\noindent
\textbf{Exact Hitting Sets: }Set systems are synonymous with hypergraphs. A \textit{hitting set} of a hypergraph $H$ is a subset $T$ of the vertex set of $H$ such that $T$ has at least one vertex from every hyperedge. If every hyperedge has exactly one element from $T$, then $T$ is called an \textit{exact hitting set}. The \textsc{Exact Hitting Set} problem is a well-studied decision problem that aims to find if a given hypergraph has an exact hitting set. It finds applications in combinatorial cryptosystems (\cite{Downey2013}) and computational biology among many others. The \textsc{Exact Hitting Set} problem is the dual of the \textsc{Exact Cover} problem which, in turn, seeks to find a set cover that covers all vertices of a hypergraph such that the number of occurrences each vertex has in the cover is exactly one. Some famous examples of the \textsc{Exact Cover} problem are sudoku, tiling dominoes, and the $n$-queens problem. The \textsc{Exact Cover} problem is a special case of the \textsc{Minimum Membership Set Cover} problem (MMSC) (\cite{Karp1972}). While the classic \textsc{Set Cover} problem seeks to find a set cover of minimum cardinality, MMSC aims to find a set cover that minimizes the maximum number of occurrences each vertex has in the cover. MMSC is known to be $\NP$-complete on arbitrary set systems (\cite{Kuhn2005}). However, for interval hypergraphs, MMSC was shown to be solvable in polynomial time by \cite{Dom2006}. If a hypergraph $H$ has an exact hitting set, we refer to $H$ as an \textit{exactly hittable hypergraph}. \cite{dhannya-NSN} have shown that a conflict-free coloring of a set of intervals is exactly a partition into sets of intervals, such that each set has an exact hitting set. This  motivates the question of characterizing those sets of intervals which have an exact hitting set.   A natural characterization is obtained by writing the hitting set linear program with one constraint per interval.  This system is totally unimodular and thus defines an integer polytope (\cite{Dom2006}). Thus, the intervals have an exact hitting set if and only if the polytope defined by the exact hitting set linear program is non-empty. 
 Further, it is possible to find if the interval hypergraph is exactly hittable in polynomial time (\cite{Dom2006}).  In this work, we consider a related graph theoretic version of this question - can we characterize the class of interval graphs that are the \textit{intersection graphs} (defined in Section \ref{sec:Prelims}) of a set of intervals that have an exact hitting set?  We refer to this class as the class of \textit{Exactly Hittable Interval Graphs} ($\ehig$). \\

\noindent
\textbf{Intersection Graphs: }The theory of graphs and hypergraphs are connected by a very well-studied notion of intersection graphs (\cite{erdos1966representation}). It is well-known that every graph $G$ is an intersection graph of some hypergraph $H$ (\cite{Harary}). $H$ is referred to as an \textit{intersection model} or a \textit{set representation} of $G$ (\cite{Gol2004,Harary}). Interestingly, certain special classes of graphs are characterized by the structure of their intersection models. For instance, \cite{gavril1974intersection} has shown that the class of chordal graphs are the intersection graphs of subtrees of a tree . 
When the hyperedges are restricted to be paths on a tree, the resulting intersection graph class is that of path chordal graphs, which is a proper subclass of the class of chordal graphs (\cite{chalopin2009planar, gavril1978recognition,leveque2009,monma1986}). These characterizations result in recognition algorithms which are very well studied.  The book by \cite{Gol2004} can be considered a pilgrimage for anyone interested in the characterization and recognition of different natural sub-classes of perfect graphs.   The recognition problem in the class of perfect graphs itself remained a fascinating open problem with a long history of results (survey in the classic book by \cite{grotschel2012}) till the Strong Perfect Graph Conjecture was proven by \cite{sgpt02}.  While many classes have efficient recognition algorithms, there are those for which the recognition problem is NP-complete.  Tolerance graphs are a sub-class of interval graphs, and the recognition problem for this class has been shown  to be NP-Complete by \cite{Mertzios2010}. The thinness of a graph, on the other hand, is a width parameter that generalizes certain properties of interval graphs. Interval graphs are exactly the graphs of thinness one. In their work, \cite{BONOMOBRABERMAN202353} have presented characterizations of 2-thin and proper 2-thin graphs as intersection graphs of rectangles in the plane, as vertex intersection graphs of paths on a grid, and by forbidden ordered patterns. Forbidden induced subgraph characterization for restricted cases of known graph classes are well-studied. For instance, even though a structural characterization by minimal forbidden induced subgraphs for the entire class of circle graphs is not known, \cite{BONOMOBRABERMAN202243} have given a characterization by minimal forbidden induced subgraphs of circle graphs, restricted to split graphs. Rectangle intersection graphs are the intersection graphs of axis-parallel rectangles in the plane. A graph is said to be a $k$-stabbable rectangle intersection graph ($k$-SRIG), if it has a rectangle intersection representation in which $k$ horizontal lines can be placed such that each rectangle intersects at least one of them. \cite{CHAKRABORTY2021354} have introduced some natural subclasses of 2-SRIG, and have shown that one of these subclasses can be recognized in linear-time if the input graphs are restricted to be triangle-free. Earlier, \cite{chakraborty2020stab} had developed a forbidden structure characterization for block graphs that are 2-ESRIG (in the case when each rectangle intersects exactly one of the k horizontal lines) and trees that are 3-ESRIG, which lead to polynomial-time recognition algorithms for these two classes of graphs. These forbidden structures are natural generalizations of asteroidal triples.

A result which has the flavour of Exact Hitting Set is in a recent paper by \cite{Bhyravarapu2021}.  They consider the problem  of coloring the vertex set of a graph with $k$ non-zero colors and one zero colour such that for each vertex $v$, there is a vertex $u$ in $N(v)$ which has a non-zero colour different from all the other vertices in $N(v)$.  This is called the $CFON^*$ colouring problem, and the goal is to find the minimum value of $k$ for which the graph has a $CFON^*$ colouring.  For $k=1$, this problem is the Exact Hitting Set problem of a set system in which the sets are the set of neighbours of each vertex.  For unit disk graphs, they show that testing if there is a $CFON^*$ coloring with one non-zero colour is NP-complete.  \\


\noindent
\textbf{Forbidden Structure Characterizations: }While a graph $G$ may be identified as an intersection graph of a structured hypergraph, characterization of $G$ based on forbidden structures has also been equally well-studied. For instance, the class of chordal graphs are characterized by the absence of induced cycles of size 4 or more (\cite{Gol2004}). 
Similarly, by the celebrated theorem of Kuratowski (\cite{West}), the class of planar graphs must not have subgraphs that are subdivisions of $K_5$ and $K_{3,3}$. Interval graphs are known to be the class of chordal graphs without an asteroidal triple as induced subgraph \cite{lekkeikerker1962}. Recall that an asteroidal triple of a graph $G$ is a set of three independent vertices such that there is path between each pair of these vertices that does not contain any vertex of the neighborhood of the third. The class of proper interval graphs is a subclass of interval graphs that do not have a $K_{1,3}$ as an induced subgraph (\cite{roberts1978graph}). Refer to Table \ref{tab:table1} for a summary of these examples. Clearly, characterization of simple graphs based on their intersection models and forbidden structures are extremely well-studied notions in defining graph classes. 

\begin{table}[ht]
\footnotesize
  \centering   
  \begin{tabular}{|p{3cm}|p{4cm}|p{6.5cm}|}
  \hline
    \textit{Graph Class} & \textit{Intersection Model} & \textit{Forbidden Structures}\\
    \hline
    Simple & An exactly hittable hypergraph & NIL\\
    \hline
Planar & Segments on a plane & Subdivisions of $K_5$ and $K_{3,3}$ \cite{West}\\
    \hline    
    Chordal & Subtrees of a tree & $C_k$, for $k \geq 4$ \cite{Gol2004}\\
    \hline
    Path chordal & Paths on a tree & \textit{List given in \cite{leveque2009}}\\
    \hline
    Interval & Subpaths on a path & $C_k$, for $k \geq 4$ and  asteroidal triple \cite{lekkeikerker1962}\\
    \hline
        Proper interval & Sets of intervals not properly \allowbreak contained in each other & $C_k$, for $k \geq 4$, asteroidal triple and $K_{1,3}$ \cite{roberts1978graph}\\
    \hline
    \textbf{Exactly Hittable \allowbreak Interval Graphs} \allowbreak (New graph class) & Exactly hittable sets of \allowbreak intervals & $C_k$, for $k \geq 4$, asteroidal triple and induced path $P_k$ which has, in its open neighbourhood, an independent set of $k+3$ vertices\\
    \hline
     \end{tabular}
     \vspace{0.2cm}
      \caption{Intersection models and forbidden structures for well-known graph classes}
  \label{tab:table1}
\end{table}

\subsection*{Our results}
\noindent
\begin{enumerate}
\item We begin our set of results with a simple extension to a well-known theorem by \cite{Harary} that every graph $G$ is the intersection graph of some hypergraph $H$.
\begin{restatable} {observation}{ObsEH}\label{obs:UGEHH}
Every simple undirected graph is the intersection graph of an exactly hittable hypergraph.
Further, if $G$ is a  chordal graph, then it is the intersection graph of an exactly hittable set of subtrees of a tree.
\end{restatable}

We present proof of this observation in Section \ref{lem:necFbIntGraph}. Further to this observation, we look at a subclass of chordal graphs, namely interval graphs, which are intersection graphs of subpaths on a path. We ask if there is an exactly hittable intersection model for every interval graph, where the intersection model consists of subpaths on a path.  Interestingly, the answer is no.\\

\item We introduce the class of \textit{Exactly Hittable Interval Graphs} ($\ehig$), which is the set of interval graphs that have an exactly hittable interval representation. A given set of intervals defines a unique interval graph, but an interval graph can have many interval representations.  We say that an interval graph is an exactly hittable interval graph if and only if it has at least one exactly hittable interval representation.

\begin{definition}[Exactly Hittable Interval Graphs]
The class of exactly hittable interval graphs is the class of interval graphs which are intersection graphs of intervals that have exact hitting sets.
\end{definition}

We present a forbidden structure characterization for $\ehig$. First, we define a family $\mathcal{F}$ of simple graphs as follows:

\begin{definition} \label{def:forbidPattern}
For each $k \geq 1$, $\mathcal{F}_k$ denotes the set of connected interval graphs whose vertex set can be partitioned into an induced path $P$ consisting of $k$ vertices and the open neighbourhood of $P$ (consisting of only those vertices which are not in $P$) which is an independent set of size  $k+3$.  Further, $\mathcal{F}$ is defined to be $\displaystyle \bigcup_{k \geq 1}\mathcal{F}_k$.
\end{definition}

Our main contribution in this paper is to prove that every graph in $\mathcal{F}$ is a forbidden structure for $\ehig$. See Fig.\ref{fig:fbIntGraphs} for examples of forbidden structures. In Fig.\ref{fig:fbIntGraphs}(i), $u$ is the induced path $P$ consisting of one vertex with an independent set of four vertices $a,b,c,d$ in its neighbourhood. Similarly, in Fig.\ref{fig:fbIntGraphs}(ii), $a$-$b$ is the induced path $P$ consisting of two vertices and $\{c,d,u,e,f\}$ is an independent set of five vertices in the neighbourhood of $P$.
 
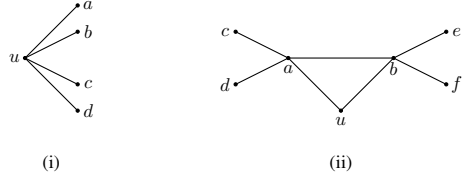
\begin{figure}[htb]
\centering
\begin{tikzpicture} [scale=0.7, every node/.style={scale=0.7}]

\draw(0,0)--(1,1) (-0.2,0) node{$u$} (1.2,1) node{$a$}; 
\draw(0,0)--(1,0.5) (1.2,0.5) node{$b$};
\draw(0,0)--(1,-0.5) (1.2,-0.5) node{$c$};
\draw(0,0)--(1,-1) (1.2,-1) node{$d$} (0.5,-2) node{(i)};

\draw[radius=0.3mm, color=black, fill=black](1,-1) circle; 
\draw[radius=0.3mm, color=black, fill=black](1,1) circle; 
\draw[radius=0.3mm, color=black, fill=black](0,0) circle; 
\draw[radius=0.3mm, color=black, fill=black](1,0.5) circle;
\draw[radius=0.3mm, color=black, fill=black](1,-0.5) circle; 

\draw(4,0.5)--(5,0); 
\draw(4,-0.5)--(5,0);
\draw(6,-1)--(5,0);
\draw(6,-1)--(7,0);
\draw(5,0)--(7,0);
\draw(8,0.5)--(7,0);
\draw(8,-0.5)--(7,0);
\draw(3.8,0.5) node{$c$};
\draw(3.8,-0.5) node{$d$};
\draw(5,-0.2) node{$a$};
\draw(7,-0.2) node{$b$};
\draw(8.2,0.5) node{$e$};
\draw(8.2,-0.5) node{$f$};
\draw(6,-1.2) node{$u$};
\draw(6,-2) node{(ii)};
\draw[radius=0.3mm, color=black, fill=black](4,0.5) circle; 
\draw[radius=0.3mm, color=black, fill=black](4,-0.5) circle; 
\draw[radius=0.3mm, color=black, fill=black](5,0) circle; 
\draw[radius=0.3mm, color=black, fill=black](6,-1) circle; 
\draw[radius=0.3mm, color=black, fill=black](7,0) circle; 
\draw[radius=0.3mm, color=black, fill=black](8,0.5) circle; 
\draw[radius=0.3mm, color=black, fill=black](8,-0.5) circle; 

\end{tikzpicture}
\caption{Two simple instances of forbidden structures}
\label{fig:fbIntGraphs}
\end{figure}
\noindent
By Definition \ref{def:forbidPattern}, 
for any $k$, the set $\mathcal{F}_k$  may contain more than one graph which are forbidden structures.  For example, both the graphs in Fig. \ref{fig:fbIntGraphsexample} belong to $\mathcal{F}_2$.
 
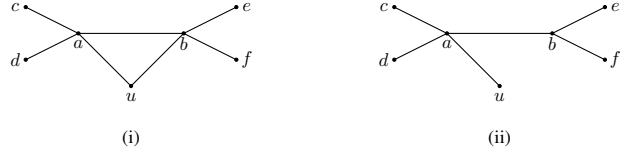
\begin{figure}[htb]
\centering
\begin{tikzpicture} [scale=0.7, every node/.style={scale=0.7}]
\draw(0,0.5)--(1,0); 
\draw(0,-0.5)--(1,0);
\draw(2,-1)--(1,0);
\draw(2,-1)--(3,0);
\draw(1,0)--(3,0);
\draw(4,0.5)--(3,0);
\draw(4,-0.5)--(3,0);
\draw(-0.2,0.5) node{$c$};
\draw(-0.2,-0.5) node{$d$};
\draw(1,-0.2) node{$a$};
\draw(3,-0.2) node{$b$};
\draw(4.2,0.5) node{$e$};
\draw(4.2,-0.5) node{$f$};
\draw(2,-1.2) node{$u$};
\draw(2,-2) node{(i)};
\draw[radius=0.3mm, color=black, fill=black](0,0.5) circle; 
\draw[radius=0.3mm, color=black, fill=black](0,-0.5) circle; 
\draw[radius=0.3mm, color=black, fill=black](1,0) circle; 
\draw[radius=0.3mm, color=black, fill=black](2,-1) circle; 
\draw[radius=0.3mm, color=black, fill=black](3,0) circle; 
\draw[radius=0.3mm, color=black, fill=black](4,0.5) circle; 
\draw[radius=0.3mm, color=black, fill=black](4,-0.5) circle; 

\draw(7,0.5)--(8,0); 
\draw(7,-0.5)--(8,0);
\draw(9,-1)--(8,0);
\draw(8,0)--(10,0);
\draw(11,0.5)--(10,0);
\draw(11,-0.5)--(10,0);
\draw(6.8,0.5) node{$c$};
\draw(6.8,-0.5) node{$d$};
\draw(8,-0.2) node{$a$};
\draw(10,-0.2) node{$b$};
\draw(11.2,0.5) node{$e$};
\draw(11.2,-0.5) node{$f$};
\draw(9,-1.2) node{$u$};
\draw(9,-2) node{(ii)};
\draw[radius=0.3mm, color=black, fill=black](7,0.5) circle; 
\draw[radius=0.3mm, color=black, fill=black](7,-0.5) circle; 
\draw[radius=0.3mm, color=black, fill=black](8,0) circle; 
\draw[radius=0.3mm, color=black, fill=black](9,-1) circle; 
\draw[radius=0.3mm, color=black, fill=black](10,0) circle; 
\draw[radius=0.3mm, color=black, fill=black](11,0.5) circle; 
\draw[radius=0.3mm, color=black, fill=black](11,-0.5) circle; 
\end{tikzpicture}
\caption{Two instances of forbidden structures in $\mathcal{F}_2$}
\label{fig:fbIntGraphsexample}
\end{figure}
\noindent It may, however, be noted that,  Fig. \ref{fig:fbIntGraphsexample}(ii) contains $K_{1,4}$ as an induced subgraph, which itself is a forbidden structure in $\mathcal{F}_1$.

\begin{restatable}{theorem}{forbidden}\label{thm:EHInterGraphs}
An interval graph $G$ is exactly hittable if and only if it does not contain any graph from the set $\mathcal{F}$ as an induced subgraph.
\end{restatable}

This theorem is proved in Section \ref{sec:IntGraphEHR}. We believe that this result is an interesting addition to the existing graph characterizations, primarily because we could not find such an equivalence elsewhere in the literature, including graph classes repositories like graphclasses.org.  \\

\item In Section \ref{sec:canonRep}, we introduce, what we refer to as, a \textit{canonical interval representation} for an interval graph. Given an interval graph $G$,  a canonical interval representation $H_G$ is an \textbf{interval hypergraph} given by $H_G=([n],\mathcal{I})$, where $[n] = \{1, \ldots, n\}$  and  $\mathcal{I} \subseteq \{ \{i, i+1, \ldots, j\} \mid i \leq j, i,j \in [n] \}$, and all intervals have distinct left endpoints and distinct right endpoints. 
Further, for each $v \in G$,  $I_v \in \mathcal{I}$ denotes the corresponding interval.  For construction of  $H_G$, we start with the well known linear ordering of maximal cliques associated with an interval graph (\cite{GilHof2011,Gol2004}). An interval representation is constructed from the ordering such that the intersection graph of this representation is isomorphic to $G$. By construction, there exists exactly one canonical interval representation for every interval graph. While the canonical representation may be of independent interest, this representation is crucial in proving Theorem \ref{thm:EHInterGraphs} in this paper.

In Section \ref{sec:canonRep}, we prove the following theorem.
\begin{restatable}{theorem}{equivalence} \label{thm:EHIHeqEHIG}
Let $G$ be an interval graph. Let $H_G$ be its canonical interval representation constructed as described in Section \ref{sec:canonRep}. Then, $G$ is exactly hittable if and only if $H_G$ is exactly hittable.
\end{restatable}

\item Given an interval graph $G$ and its canonical interval representation $H_G$, we show that the algorithm by \cite{Dom2006} to solve the MMSC problem in interval hypergraphs can be used to recognize $\ehig$. We present the details in Section \ref{sec:algoEHIG}.\\

\item We show that the class $\ehig$ is positioned between the class of proper interval graphs and the class of interval graphs in the  containment hierarchy of graph classes.

\begin{restatable}{theorem}{hierarchy} \label{thm:Hierarchy}
Proper interval graphs $\subset$ \textsf{EHIG} $\subset$ Interval Graphs.
\end{restatable}


\noindent
The proof of the second part of the above theorem follows from the definition of $\ehig$. We prove the first part of the containment relationship in Section \ref{subsec:PropSubSetEHIG}. Interestingly, the smallest forbidden structure of $\ehig$ is $K_{1,4}$ whereas that of the class of proper interval graphs is $K_{1,3}$.
\end{enumerate}
\subsection{Preliminaries}\label{sec:Prelims}

\begin{definition}[Intersection Graphs]
Given a set system $\mathcal{X}=(\mathcal{U},\mathcal{S})$, \textit{the intersection graph} $G(\mathcal{X})$ of sets in $\mathcal{X}$ is the simple graph obtained as follows. For every set $S \in \mathcal{S}$, there exists a vertex $v_S \in G$. An edge $(v_{S_i},v_{S_j})$ occurs in $G$ if and only if there exists two sets $S_i,S_j \in \mathcal{F}$ such that $S_i \cap S_j \neq \emptyset$. The family $\mathcal{S}$ is called a \textit{set representation} of the graph G. A set representation is also referred to as an \textit{intersection model} (\cite{Gol2004}, \cite{Harary}). 
\end{definition}

A hypergraph $H=(\mathcal{V},\mathcal{E})$ is a graph theoretic representation of a set system $\mathcal{X}=(\mathcal{U},\mathcal{S})$, where the set $\mathcal{V}$ corresponds to $\mathcal{U}$ and the set $\mathcal{E}$ corresponds to $\mathcal{S}$. The set $\mathcal{V}$ contains \textit{vertices} of hypergraph $H$ and the set $\mathcal{E}$ contains \textit{hyperedges}. In the intersection graph $G$, for every hyperedge $E \in \mathcal{E}$, there exists a vertex $v_E \in G$. An edge $(v_{E_i},v_{E_j})$ occurs in $G$ if and only if the hyperedges $E_i$ and $E_j$ have a non-empty intersection. 

\begin{definition}[Interval Graphs]
 A graph $G=(V,E)$ is an \textit{interval graph} if there exists an assignment of intervals on the real line to each vertex $v \in V(G)$ such that for each edge $(u,v)$ in $G$, the associated intervals $I(u)$ and $I(v)$ have a non-empty intersection. The set of intervals $\{I(v)\}_{v\in V(G)}$ is an interval representation or intersection model of $G$. 
\end{definition}

\noindent
\textbf{Open and Closed neighborhoods:} For a vertex $v$ in a graph $G = (V,E)$, the open neighborhood of $v$ in $G$, denoted by $N(v)$, is the set $\{u \in V \mid \{u,v\} \in E\}$ and the closed neighborhood of $v$ in $G$, denoted by $N[v]$, is the set  $N(v) \cup \{v\}$.  
\begin{definition} \cite{CS2012} 
An \textbf{interval hypergraph} is any hypergraph $H=([n],\mathcal{I})$, where $[n] = \{1, \ldots, n\}$  and $\mathcal{I} \subseteq \{ \{i, i+1, \ldots, j\} \mid i \leq j, i,j \in [n] \}$. 
\end{definition}
Each hyperedge in $\mathcal{I}$ is a set of consecutive integers, which we call an \textit{interval}. In an interval $I = \{i,i+1,\ldots,j\}$, $i$ and $j$ are the \textit{left} and \textit{right endpoints} of $I$ respectively, which we denote by $l(I)$ and $r(I)$, respectively. We use $\mathcal{V}(H)$ (or simply $\mathcal{V}$) and $\mathcal{I}(H)$ (or simply $\mathcal{I}$) to denote the vertex set and the hyperedge set, respectively, of an interval hypergraph $H$. An interval hypergraph is said to be \textit{proper} if no interval is contained in another interval. If, for an interval graph $G$, there exists an interval representation in which no interval is properly contained inside another interval, then $G$ is a \textit{proper interval graph}. 

An interval graph is characterized by the existence of a linear ordering of its maximal cliques. In Section \ref{sec:IntGraphEHR}, we use the following characterization to obtain an exactly hittable interval representation for an interval graph, if such a representation exists. 

\begin{theorem}[\cite{GilHof2011}] 
\label{thm:linord}
The maximal cliques of an interval graph $G$ can be linearly ordered such that, for every vertex $x$ of $G$, the maximal cliques containing $x$ occur consecutively.
\end{theorem}

\noindent The class of interval graphs is a subfamily of the class of \textit{chordal graphs}, which, in turn, is a subfamily of the class of \textit{perfect graphs}. A chordal graph is a simple graph that does not contain any induced cycle of size $\geq 4$ (\cite{Gol2004}). Chordal graphs are known to be intersection graphs of subtrees of a tree (\cite{gavril1974intersection}).
A clique tree $T$ of a graph $G$ is a tree with the maximal cliques of $G$ as nodes, such that for every vertex $v$ of $G$, the maximal cliques containing 
$v$ induce a subtree $T(v)$ in $T$. In fact, chordal graphs are exactly the graphs that admit a clique tree (\cite{mckee1999}). A clique tree is also known as a \textit{tree decomposition} of a graph. \\


\noindent
\textbf{Note:} We draw the reader's attention to the distinction between \textit{interval hypergraphs} and \textit{interval graphs}, and {\em proper interval hypergraphs} and {\em proper interval graphs}, as these are used extensively throughout the paper. Furthermore, recall that an interval graph is an Exactly Hittable Interval Graph if it has an intersection model, made of intervals, that has an exact hitting set.  On the other hand, an Exactly Hittable Interval Hypergraph is one that has an exact hitting set.  

\begin{observation} \label{obs:imp1}
Since our goal is to characterize interval graphs that have an exactly hittable interval representation, we assume without loss of generality that, in the graph $G$, for every sequence of consecutive maximal cliques in a linear ordering, there is at most one vertex which starts and ends in this sequence.
\end{observation}

\noindent
Indeed, if a given graph violates this property and there are two or more vertices that start at the same clique and end at the same clique in a sequence, then we retain only one of those vertices. The justification for this assertion is that if the resulting graph has an exactly hittable interval representation, so does the original graph. \\

\noindent
\textbf{Notations}: All other definitions and notations on simple graphs, used throughout this paper, have been taken from \cite{West}.

\section{A Canonical Interval Representation} \label{sec:canonRep}
In this section, we obtain \textit{a canonical interval representation} $H_G$ of a given interval graph $G$. The canonical interval representation is nothing but a special intersection model of $G$. Consequently, the intersection graph of intervals in $H_G$ is isomorphic to $G$. The construction follows a well-defined set of steps with the result that every interval graph has a unique canonical interval representation. The canonical representation $H_G$ is obtained by \textit{stretching} intervals so that all intervals have distinct left endpoints and distinct right endpoints. In other words, no pair of intervals start at the same point or end at the same point. 
The canonical interval representation
is crucial to the proof of our main result in Section \ref{sec:IntGraphEHR}.\\

\noindent
\textbf{Outline: }The starting point of this construction is to use the well known linear ordering of maximal cliques associated with an interval graph (\cite{Gol2004}) (refer Theorem \ref{thm:linord}). 
\noindent
Fig. \ref{fig:IntGraphMinRep} gives an illustration of how to obtain the canonical interval representation of an interval graph. 
Let $G=(V,E)$ be the given interval graph. Let $\mathcal{O}=\{Q_1, Q_2 \ldots Q_t\}$ be a linear ordering of maximal cliques in $G$. For each $v \in V(G)$, let the interval representation of $G$ obtained from ${\cal O}$ be $I(v) = [l(v), r(v)]$, where $l(v)$ is the index of the leftmost clique in ${\cal O}$ that contains $v$, and $r(v)$ is the index of the rightmost clique containing $v$. Let $\mathcal{I}' = \{I(v) \mid v \in V(G)\}$. To construct the canonical interval representation, we associate a gadget $D_i$ with maximal clique $Q_i$, for $1 \leq i \leq t$. For every maximal clique $Q_i$, we look at $D_i$ and stretch those intervals in $\mathcal{I}'$ that either start at $i$ or end at $i$. Intuitively, we can think of $I(v)$ as being {\em stretched  to the left} if $l(v)=i$ and as being {\em stretched to the right} if $r(v)=i$. Inside gadget $D_i$, there is a point, which we denote by $z_i$, with the following property: any interval for which $l(v)=i$, starts at $z_i$ or to the left of $z_i$ and any interval for which $r(v)=i$, ends at $z_i$ or to the right of $z_i$. We refer to $z_i$ as the \textit{zero-point} of gadget $D_i$. The exact construction of stretched intervals is detailed in the subsequent paragraphs.

The gadgets $D_1,D_2 \ldots, D_t$ are arranged in the same order as that of the maximal cliques in ${\cal O}$. Further,  for each $v \in V(G)$, the stretched interval associated with $I(v)$ has $D_{l(v)}$ as its left-most gadget and $D_{r(v)}$ as its rightmost gadget.  To complete the construction, between each pair of consecutive gadgets, we add an additional point, and we refer to these points as intermediate points.  These points play a crucial role in our characterization of $\ehig$s in Section \ref{subsec:computeExactHittingSet}. 
The stretched interval of $I(v)$ contains all these additional points between consecutive gadgets in the ordered set $\{D_{l(v)}, D_{l(v)+1}, \ldots, D_{r(v)}\}$.  Let $H_G = (\mathcal{V},\mathcal{I})$ denote the canonical interval hypergraph thus obtained.  $\mathcal{V}$ is the set of all points internal to the gadgets (defined below) and the $t-1$ additional points between consecutive gadgets (as described above).  The intervals in ${\cal I}$ are the stretched intervals corresponding to each interval in ${\cal I'}$.  We now describe the gadget $D_i$ associated with maximal clique $Q_i$, $1 \leq i \leq t$.\\

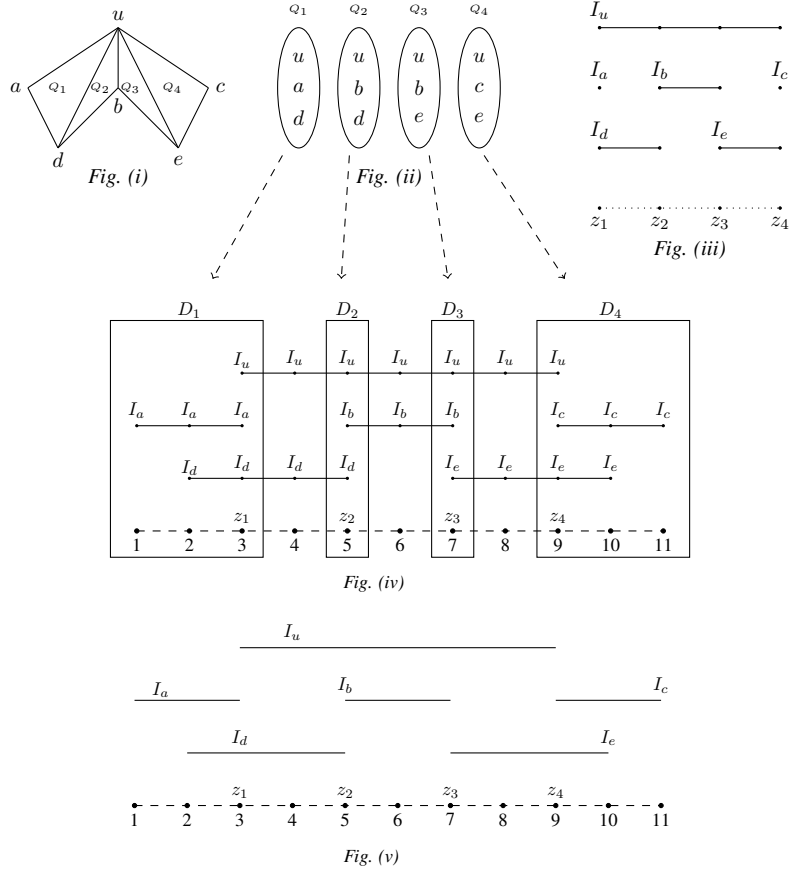
\begin{figure}[ht]
\centering
\begin{tikzpicture}[scale=0.8, every node/.style={scale=0.8}]
\draw(4.5,0)--(6,1); 
\draw(4.5,0)--(5,-1);
\draw(6,1)--(6,0);
\draw(6,1)--(7.5,0);
\draw(6,1)--(7,-1);
\draw(6,0)--(5,-1);
\draw(6,0)--(7,-1);
\draw(7,-1)--(7.5,0);
\draw(6,1)--(5,-1);
\draw(6,1.2) node{$u$};
\draw(4.3,0) node{$a$};
\draw(5,-1.2) node{$d$};
\draw(6,-0.3) node{$b$};
\draw(7.7,0) node{$c$};
\draw(7,-1.2) node{$e$} (6,-1.5) node{\textit{Fig. (i)}};

\draw(5,0) node{\tiny{$Q_1$}};
\draw(5.7,0) node{\tiny{$Q_2$}};
\draw(6.2,0) node{\tiny{$Q_3$}};
\draw(6.9,0) node{\tiny{$Q_4$}};

\draw(9,0.5) node{$u$};
\draw(10,0.5) node{$u$};
\draw(11,0.5) node{$u$};
\draw(12,0.5) node{$u$};
\draw(9,0) node{$a$};
\draw(10,0) node{$b$};
\draw(11,0) node{$b$};
\draw(12,0) node{$c$};
\draw(9,-0.5) node{$d$};
\draw(10,-0.5) node{$d$};
\draw(11,-0.5) node{$e$};
\draw(12,-0.5) node{$e$};

\draw(10.5,-1.5) node{\textit{Fig. (ii)}};

\draw(9,1.3) node{\tiny{$Q_1$}};
\draw(10,1.3) node{\tiny{$Q_2$}};
\draw(11,1.3) node{\tiny{$Q_3$}};
\draw(12,1.3) node{\tiny{$Q_4$}};

\draw(10.5,-2.5) node{}; 

\draw (9,0) ellipse (0.35cm and 1cm);
\draw (10,0) ellipse (0.35cm and 1cm);
\draw (11,0) ellipse (0.35cm and 1cm);
\draw (12,0) ellipse (0.35cm and 1cm);

\draw[radius=0.2mm, fill=black](14,1) circle node[above=0.02cm] {$I_u$};
\draw[radius=0.2mm, fill=black](15,1) circle;
\draw[radius=0.2mm, fill=black](16,1) circle;
\draw[radius=0.2mm, fill=black](17,1) circle;
\draw(14,1)--(17,1);

\draw[radius=0.2mm, fill=black](14,0) circle node[above=0.02cm] {$I_a$};
\draw[radius=0.2mm, fill=black](15,0) circle node[above=0.02cm] {$I_b$};
\draw[radius=0.2mm, fill=black](16,0) circle;
\draw[radius=0.2mm, fill=black](17,0) circle node[above=0.02cm] {$I_c$};
\draw[radius=0.2mm, fill=black](14,-1) circle node[above=0.02cm] {$I_d$};
\draw[radius=0.2mm, fill=black](15,-1) circle;
\draw[radius=0.2mm, fill=black](16,-1) circle node[above=0.02cm] {$I_e$};
\draw[radius=0.2mm, fill=black](17,-1) circle;
\draw(15,0)--(16,0); 
\draw(14,-1)--(15,-1); 
\draw(16,-1)--(17,-1); 

\draw[dotted](14,-2)--(17,-2);
\draw[radius=0.2mm, fill=black](14,-2) circle node[below=0.02cm] {$z_1$};
\draw[radius=0.2mm, fill=black](15,-2) circle node[below=0.02cm] {$z_2$};
\draw[radius=0.2mm, fill=black](16,-2) circle node[below=0.02cm] {$z_3$};
\draw[radius=0.2mm, fill=black](17,-2) circle node[below=0.02cm] {$z_4$};

\draw(15.5,-2.7) node{\textit{Fig. (iii)}};

\node (Q1) at (8.8,-1){};
\node (D1) at (7.5,-3.3){};
\draw[dashed,->]  (Q1) to (D1);

\node (Q2) at (9.85,-1){}; 
\node (D2) at (9.7,-3.3){}; 
\draw[dashed,->]  (Q2) to (D2);

\node (Q3) at (11.15,-1){}; 
\node (D3) at (11.5,-3.3){}; 
\draw[dashed,->]  (Q3) to (D3);

\node (Q4) at (12,-1){}; 
\node (D4) at (13.5,-3.3){}; 
\draw[dashed,->]  (Q4) to (D4);

\end{tikzpicture}

\begin{tikzpicture}[scale=0.7, every node/.style={scale=0.7}]

\draw[radius=0.2mm, fill=black](2,-1) circle;
\draw[radius=0.2mm, fill=black](3,1) circle;

\draw(1.5,0.2);
\draw(3,1.25) node{$I_u$};
\draw(2,-0.8) node{$I_d$};
\draw(2,2.2) node{$D_1$};
\draw(5,2.2) node{$D_2$};
\draw(7,2.2) node{$D_3$};
\draw(10,2.2) node{$D_4$};

\draw[radius=0.2mm, fill=black](4,1) circle node[above=0.02cm] {$I_u$};
\draw[radius=0.2mm, fill=black](5,1) circle node[above=0.02cm] {$I_u$};
\draw[radius=0.2mm, fill=black](6,1) circle node[above=0.02cm] {$I_u$};
\draw[radius=0.2mm, fill=black](7,1) circle node[above=0.02cm] {$I_u$};
\draw[radius=0.2mm, fill=black](8,1) circle node[above=0.02cm] {$I_u$};
\draw[radius=0.2mm, fill=black](9,1) circle node[above=0.02cm] {$I_u$};

\draw[radius=0.2mm, fill=black](1,0) circle node[above=0.02cm] {$I_a$};
\draw[radius=0.2mm, fill=black](2,0) circle node[above=0.02cm] {$I_a$};
\draw[radius=0.2mm, fill=black](3,0) circle node[above=0.02cm] {$I_a$};

\draw[radius=0.2mm, fill=black](4,-1) circle node[above=0.02cm] {$I_d$};
\draw[radius=0.2mm, fill=black](3,-1) circle node[above=0.02cm] {$I_d$};
\draw[radius=0.2mm, fill=black](5,-1) circle node[above=0.02cm] {$I_d$};

\draw[radius=0.2mm, fill=black](5,0) circle node[above=0.02cm] {$I_b$};
\draw[radius=0.2mm, fill=black](6,0) circle node[above=0.02cm] {$I_b$};
\draw[radius=0.2mm, fill=black](7,0) circle node[above=0.02cm] {$I_b$};

\draw[radius=0.2mm, fill=black](7,-1) circle node[above=0.02cm] {$I_e$};
\draw[radius=0.2mm, fill=black](8,-1) circle node[above=0.02cm] {$I_e$};
\draw[radius=0.2mm, fill=black](9,-1) circle node[above=0.02cm] {$I_e$};
\draw[radius=0.2mm, fill=black](10,-1) circle node[above=0.02cm] {$I_e$};

\draw[radius=0.2mm, fill=black](11,0) circle node[above=0.02cm] {$I_c$};
\draw[radius=0.2mm, fill=black](9,0) circle node[above=0.02cm] {$I_c$};
\draw[radius=0.2mm, fill=black](10,0) circle node[above=0.02cm] {$I_c$};

\draw (0.5,2) rectangle (3.4,-2.5);
\draw (4.6,2) rectangle (5.4,-2.5);
\draw (6.6,2) rectangle (7.4,-2.5);
\draw (8.6,2) rectangle (11.5,-2.5);

\draw[radius=0.4mm, color=black, fill=black](1,-2) circle node[below=0.02cm]{1};
\draw[radius=0.4mm, color=black, fill=black](2,-2) circle node[below=0.02cm]{2};
\draw[radius=0.4mm, color=black, fill=black](3,-2) circle node[below=0.02cm]{3} node[above=0.02cm] {$z_1$};
\draw[radius=0.4mm, color=black, fill=black](4,-2) circle node[below=0.02cm]{4};
\draw[radius=0.4mm, color=black, fill=black](5,-2) circle node[below=0.02cm] {5} node[above=0.02cm] {$z_2$};
\draw[radius=0.4mm, color=black, fill=black](6,-2) circle node[below=0.02cm]{6}; 
\draw[radius=0.4mm, color=black, fill=black](7,-2) circle node[below=0.02cm]{7} node[above=0.02cm] {$z_3$};
\draw[radius=0.4mm, color=black, fill=black](8,-2) circle node[below=0.02cm]{8};
\draw[radius=0.4mm, color=black, fill=black](9,-2) circle node[below=0.02cm]{9} node[above=0.02cm] {$z_4$};;
\draw[radius=0.4mm, color=black, fill=black](10,-2) circle node[below=0.02cm]{10};
\draw[radius=0.4mm, color=black, fill=black](11,-2) circle node[below=0.02cm]{11};
\draw[dashed](1,-2)--(11,-2); 

\draw(3,1)--(9,1); 
\draw(1,0)--(3,0); 
\draw(5,0)--(7,0); 
\draw(2,-1)--(5,-1); 
\draw(7,-1)--(10,-1); 
\draw(9,0)--(11,0); 
\draw(5.5,-3) node{\textit{Fig. (iv)}};
\draw(4.5,-3.5) node{}; 
\end{tikzpicture}

\begin{tikzpicture}[scale=0.7, every node/.style={scale=0.7}]
\draw(3,1)--(9,1); 
\draw(1,0)--(3,0); 
\draw(5,0)--(7,0); 
\draw(2,-1)--(5,-1); 
\draw(7,-1)--(10,-1); 
\draw(9,0)--(11,0); 

\draw(1.5,0.2) node{$I_a$};
\draw(3,-1) node[above=0.02cm] {$I_d$};
\draw(4,1) node[above=0.02cm] {$I_u$};
\draw(5,0) node[above=0.02cm] {$I_b$};
\draw(10,-1) node[above=0.02cm] {$I_e$};
\draw(11,0) node[above=0.02cm] {$I_c$};

\draw[radius=0.4mm, color=black, fill=black](1,-2) circle node[below=0.02cm]{1};
\draw[radius=0.4mm, color=black, fill=black](2,-2) circle node[below=0.02cm]{2};
\draw[radius=0.4mm, color=black, fill=black](3,-2) circle node[below=0.02cm]{3} node[above=0.02cm] {$z_1$};
\draw[radius=0.4mm, color=black, fill=black](4,-2) circle node[below=0.02cm]{4};
\draw[radius=0.4mm, color=black, fill=black](5,-2) circle node[below=0.02cm] {5} node[above=0.02cm] {$z_2$};
\draw[radius=0.4mm, color=black, fill=black](6,-2) circle node[below=0.02cm]{6}; 
\draw[radius=0.4mm, color=black, fill=black](7,-2) circle node[below=0.02cm]{7} node[above=0.02cm] {$z_3$};
\draw[radius=0.4mm, color=black, fill=black](8,-2) circle node[below=0.02cm]{8};
\draw[radius=0.4mm, color=black, fill=black](9,-2) circle node[below=0.02cm]{9} node[above=0.02cm] {$z_4$};;
\draw[radius=0.4mm, color=black, fill=black](10,-2) circle node[below=0.02cm]{10};
\draw[radius=0.4mm, color=black, fill=black](11,-2) circle node[below=0.02cm]{11};
\draw[dashed](1,-2)--(11,-2); 
\draw(5.5,-3) node{\textit{Fig. (v)}};
\end{tikzpicture} 

\caption{\scriptsize{Construction of Canonical Interval Representation \textit{(i)} Interval Graph $G$ with its maximal cliques $Q_1,Q_2,Q_3,Q_4$ \textit{(ii)} Linear ordering of maximal cliques $\mathcal{O} = \{Q_1,Q_2,Q_3,Q_4\}$  \textit{(iii)} Interval representation of $G$ obtained from $\mathcal{O}$ \textit{(iv)} Gadgets $D_1$ to $D_4$ \textit{(v)} Canonical interval representation for $G$ }}
\label{fig:IntGraphMinRep}
\end{figure}%

\noindent
\textbf{Construction of the gadget $D_i$ for maximal clique $Q_i$:} 
Let $\{I(v_1),I(v_2), $ $\ldots, I(v_m)\}$ be the ordered set of intervals such that for each $1 \leq k \leq m$, $l(v_k) = i$ and $r(v_k) > r(v_j)$ whenever $1 \leq k < j \leq m$.   In other words, the ordered set  considers the intervals whose left endpoint is $i$ in descending order of their right endpoints. Then, for each $1 \leq k \leq m$,  the left endpoint of the interval of $I(v_k)$ is stretched $k-1$ points to the left. By this, $l(v_1)$ is kept at $z_i$ itself, as no stretching is done on it 
(see Fig. \ref{fig:gadgetImage}).

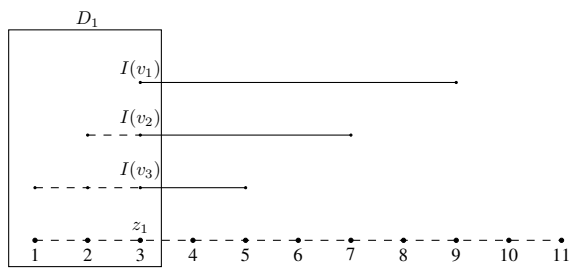
\begin{figure}[ht]
\centering
\begin{tikzpicture}[scale=0.7, every node/.style={scale=0.7}]

\draw[radius=0.2mm, fill=black](2,-1) circle;
\draw[radius=0.2mm, fill=black](3,1) circle;

\draw(1.5,0.2);
\draw(2,2.2) node{$D_1$};
\draw (0.5,2) rectangle (3.4,-2.5);

\draw(3,1.25) node{$I(v_1)$};
\draw[radius=0.2mm, fill=black](9,1) circle node[above=0.02cm] {};

\draw[radius=0.2mm, fill=black](7,0) circle node[above=0.02cm] {};
\draw[radius=0.2mm, fill=black](3,0) circle node[above=0.02cm] {$I(v_2)$};
\draw[radius=0.2mm, fill=black](2,0) circle node[above=0.02cm] {};

\draw[radius=0.2mm, fill=black](3,-1) circle node[above=0.02cm] {$I(v_3)$};
\draw[radius=0.2mm, fill=black](5,-1) circle node[above=0.02cm] {};
\draw[radius=0.2mm, fill=black](1,-1) circle node[above=0.02cm] {};

\draw[radius=0.4mm, color=black, fill=black](1,-2) circle node[below=0.02cm]{1};
\draw[radius=0.4mm, color=black, fill=black](2,-2) circle node[below=0.02cm]{2};
\draw[radius=0.4mm, color=black, fill=black](3,-2) circle node[below=0.02cm]{3} node[above=0.02cm] {$z_1$};
\draw[radius=0.4mm, color=black, fill=black](4,-2) circle node[below=0.02cm]{4};
\draw[radius=0.4mm, color=black, fill=black](5,-2) circle node[below=0.02cm] {5};
\draw[radius=0.4mm, color=black, fill=black](6,-2) circle node[below=0.02cm]{6}; 
\draw[radius=0.4mm, color=black, fill=black](7,-2) circle node[below=0.02cm]{7};
\draw[radius=0.4mm, color=black, fill=black](8,-2) circle node[below=0.02cm]{8};
\draw[radius=0.4mm, color=black, fill=black](9,-2) circle node[below=0.02cm]{9};
\draw[radius=0.4mm, color=black, fill=black](10,-2) circle node[below=0.02cm]{10};
\draw[radius=0.4mm, color=black, fill=black](11,-2) circle node[below=0.02cm]{11};
\draw[dashed](1,-2)--(11,-2); 

\draw(3,1)--(9,1); 
\draw(3,0)--(7,0); 
\draw[dashed](2,0)--(3,0); 
\draw(3,-1)--(5,-1); 
\draw[dashed](1,-1)--(3,-1); 
\end{tikzpicture}
\caption{\scriptsize{\textit{Stretching} intervals to the left}}
\label{fig:gadgetImage}
\end{figure}%

On the integer line, the left end point of $I(v_k) \in D_1$, which is the most left stretched interval in $D_1$, is taken as point 1.  
We next consider those intervals $I(v)$ such that $r(v)=i$.  Let $\{I(v_1),I(v_2), \ldots, I(v_m)\}$ be the ordered set of intervals such that for each $1 \leq k \leq t$, $r(v_k) = i$ and $l(v_k) < l(v_j)$ whenever $1 \leq k < j \leq m$.   In other words, the ordered set  considers the intervals whose right endpoint is $i$ in ascending order of their left endpoints. Then, for each $1 \leq k \leq m$,  the right endpoint of the interval of $I(v_k)$ is stretched $k-1$ points to the right. 
On the integer line, the right endpoint of the stretched interval of $I(v_k)$ would be $z_i+k-1$. This completes the description of the gadget $D_i$. Note that for $I(v)$ in ${\cal I}$, the stretched interval is stretched to the left only in the leftmost gadget in which it is present, and it is stretched to the right in the rightmost gadget in which it is present. By construction, no two intervals share the same left endpoint and the same right endpoint.

\begin{lemma} \label{lem:canonEqIntGraph} 
Let $H_G$ be the canonical interval representation of graph $G$ as constructed using the above procedure. Then, $G$ is isomorphic to the intersection graph of intervals in $H_G$.
\end{lemma}
\begin{proof}
The gadgets $D_1, \ldots, D_t$ are arranged in the same order as the maximal cliques in the ordered set $\mathcal{O} = \{Q_1, Q_2 \ldots Q_t\}$. For each $v \in G$, the starting gadget (and the ending gadget) of interval $I(v)$ and the starting maximal clique (and the ending maximal clique) of vertex $v$ in $\mathcal{O}$ are the same by construction. Further, $I(v)$ contains all the points in the intervening gadgets between the starting and ending gadgets of $I(v)$ just as $v$ occurs in all the intervening maximal cliques between the starting and ending maximal cliques to which $v$ belongs to. It follows that $I(u)$ and $I(v)$ intersect if and only if the corresponding stretched intervals have a non-empty intersection. Thus the intersection graph of intervals in $H_G$ is isomorphic to $G$. 
\end{proof}
\noindent
\section{Exactly Hittable Interval Graphs} \label{sec:IntGraphEHR}
Characterizing simple graphs as intersection graphs is a well-pursued line of study in graph theory. \cite{Harary} had presented results on this problem in his book. We address the question of when a simple graph is the intersection graph of an exactly hittable hypergraph. We modify the proof given by Harary to answer this question. In addition, we present similar results for the class of chordal graphs (refer to Section \ref{sec:Prelims} for definition). We recall and prove Observation \ref{obs:UGEHH} about arbitrary graphs and arbitrary chordal graphs. 

\noindent
\ObsEH*
\begin{proof}
The proof of the first statement is based on a slight modification to the intersection model constructed from $G$ in Theorem 2.5 in the book by \cite{Harary}. Let $H = (\mathcal{V},\mathcal{E})$ be the intersection model constructed as follows. The universe $\mathcal{V}$ of the hypergraph is $V(G) \cup E(G)$.  The set $\mathcal{E}$ contains a hyperedge $E_v$ for each vertex $v \in V(G)$, and $E_v$ contains all the edges incident on $v$ and the element $v$. Clearly, the intersection graph of $H$ is isomorphic to $G$ and $V(G)$ is an exact hitting set of $H$. \\
\noindent
The proof of the second statement, which is for a chordal graph $G$, is similar and is as follows.  Since $G$ is a chordal graph let it be isomorphic to the intersection graph of some subtrees of a tree $T$. In particular, let $T$ be the clique tree of the chordal graph $G$ (\cite{Gol2004}). Let $\{T_v \mid v \in V(G)\}$ be the set of subtrees in $T$, where $T_v$ is the subtree associated with $v$ and the tree nodes in $T_v$ correspond to those maximal cliques in $G$ which contain the vertex $v$.  We modify $T$ to get $T'$ by adding $n = |V(G)|$ new nodes, each corresponding to a vertex in $V(G)$.  For each $v \in V(G)$, the new node corresponding to $v$ is made adjacent in $T$ to some node in $T_v$.  The resulting tree is $T'$ and $T'_v$ is the subtree of $T'$ consisting of $T_v$ and the new node corresponding to $v$.  Clearly, the newly added nodes form an exact hitting set of the set $\{T'_v \mid v \in V(G)\}$ in $T'$, and the intersection graph of the subtrees $\{T'_v \mid v \in G\}$ is the same as $G$.  
 \end{proof}

Interestingly, not every interval graph has an exactly hittable interval representation. 
In this paper, we present a forbidden structure characterization for the class of interval graphs that have an exactly hittable interval representation. 
In this section, we prove that every graph in $\mathcal{F}$ (see Definition 1) is a forbidden structure for $\ehig$. First, we state and prove one direction of Theorem \ref{thm:EHInterGraphs}. 

We use the following notations throughout the section. $H'$ denotes an interval representation of $G$. We denote the open neighbourhood of vertex $v$ by $N(v)$. $N(P)$ denotes open neighbourhood of all vertices in path $P$, excluding the vertices in $P$. $\mathcal{I}(P)$  denotes the set of intervals in $H'$ corresponding to vertices in path $P$, $X_{N(P)}$  denotes the set of independent vertices in $N(P)$ and $\mathcal{I}(X_{N(P)})$ denotes set of intervals in $H'$ corresponding to $X_{N(P)}$.

\begin{lemma} \label{lem:necFbIntGraph}
Let $G$ be an interval graph. Let $F \in \mathcal{F}$ be any forbidden structure. If $G$ contains $F$ as an induced subgraph, then $G$ is not an Exactly Hittable Interval Graph.
\end{lemma}
\begin{proof}
Our proof is by contradiction. Let $H'$ be any exactly hittable interval representation of $G$  and let $G$ contain $F\in \mathcal{F}$ as an induced subgraph.  Let $F$ contain $P$, an induced path of length $k$ in $G$ that has an independent set of at least $k+3$ vertices in its  neighbourhood. Let $S$ be an exact hitting set of $H'$. Recall that $\mathcal{I}(P)$ denotes the set of intervals in $H'$ corresponding to vertices in path $P$. By our assumption that $G$ contains $F$, the number of intervals in $\mathcal{I}(X_{N(P)})$ is at least $k+3$. Hence $|\mathcal{I}(X_{N(P)}) \cap S| \geq k+3$. Since $X_{N(P)}$ is an independent set,  there can be at most two intervals in $\mathcal{I}(X_{N(P)})$ that have at least one endpoint each outside the union of intervals in $\mathcal{I}(P)$ - one on either side of $P$. Therefore, even if these two intervals in $\mathcal{I}(X_{N(P)})$ are hit outside the intervals in $\mathcal{I}(P)$ at either ends, the remaining $k+1$ independent intervals have to be hit inside the union of intervals in $\mathcal{I}(P)$. Hence $|\mathcal{I}(P) \cap S| \geq k+1$. But there are only $k$ intervals inside $\mathcal{I}(P)$. Therefore, by the pigeonhole principle, at least one interval among the intervals in $\mathcal{I}(P)$ has to be hit more than once. Thus $S$ cannot be an exact hitting set of $H'$. We have arrived at a contradiction to the assumption that  $H'$ is  exactly hittable. Since we started with an arbitrary exactly hittable representation and arrived at a contradiction, we conclude that $G$ is not exactly hittable.
\end{proof}


\noindent Now, we prove the other direction of Theorem \ref{thm:EHInterGraphs}, i.e, an interval graph $G$ which contains no graph from the set $\mathcal F$ as an induced subgraph is exactly hittable. 
\noindent Let $\mathcal{O} = \{Q_1, Q_2 \ldots Q_t\}$ be a linear ordering of maximal cliques in $G$ (refer Theorem \ref{thm:linord} and Section \ref{sec:canonRep}). Let $H_G$ be the canonical interval representation of $G$ obtained from $\mathcal{O}$. We use the following notations in this section. We denote a minimum clique cover of the neighbourhood of a vertex $v$, which is formed by the minimum number of maximal cliques in $\mathcal{O}$, by $C(N[v])$. Recall that a clique cover for a vertex set $S$ is a set of cliques such that each vertex in $S$ appears in at least one clique. Note that such a clique cover exists. 

\noindent
We prove a simple observation here.
\begin{observation}\label{obs:lastclique}
If $Q_i\ldots Q_j, ~i,j\in[1,t], i\leq j$ denote the maximal cliques containing vertex $v\in V$, then $Q_j\in C(N[v])$.
\end{observation}

\begin{proof} 
We prove this by contradiction. Let us assume that $Q_j \notin C(N[v])$.  As $Q_j\neq Q_{j-1}$,  there exists a vertex $u$ in $Q_j$ which is not in $Q_{j-1}$.  It follows that $u$ is not contained in any maximal cliques that occur before 
$Q_{j-1}$ in $\mathcal{O}$ since  the maximal cliques containing a vertex occur consecutively in the linear ordering of maximal cliques of an interval graph.
Therefore, if  $Q_j \notin C(N[v])$, then $u$ is not covered. It contradicts the fact that $ C(N[v])$ is a  clique cover of $N[v]$. It follows that $Q_j\in C(N[v])$. 
\end{proof}

From now on, when we refer to a minimum clique cover of the input graph, we mean a minimum clique cover formed by the minimum number of maximal cliques in $\mathcal{O}$ unless specified otherwise. Let $| C(N[v]) |$ denote the number of cliques in $C(N[v])$. Similarly, we denote a minimum clique cover of vertices in the maximal cliques  $Q_i$ to $Q_j$ in the ordering $\mathcal{O}$, $i < j$, by $C(Q_i, \ldots, Q_j)$. 
 
Our proof is based on the structural properties of a path $P$ in $G$, the construction of which is presented in Algorithm \ref{algo:constructP}.  The structural properties of path $P$ are proved as lemmas later in the section.

\subsection*{Outline of Algorithm \ref{algo:constructP}:}

\noindent We construct an induced path $P$ which contains a minimal set of vertices from graph $G$. The vertices in path $P$ are selected such that every maximal clique in $\mathcal{O}$ has a non-empty intersection with path $P$. Further, we incrementally construct a clique cover  of $G$ by taking the clique cover of the closed neighbourhood of each of the individual vertices in $P$. 

 \begin{figure}[b]
\centering
\begin{tikzpicture}[scale=0.8, every node/.style={scale=0.8}]
\draw(0.3,2)--(4,2);
\draw(3.1,0.5)--(9,0.5);

\draw (1,1) ellipse (0.5cm and 1.75cm);
\draw[very thick,dotted] (2,1)--(3,1);
\draw (4,1) ellipse (0.5cm and 1.75cm);

\draw (6,1) ellipse (0.5cm and 1.75cm);
\draw[very thick,dotted] (7,1)--(8,1);
\draw (9,1) ellipse (0.5cm and 1.75cm);
\draw (11,1) ellipse (0.5cm and 1.75cm);
\draw[very thick,dotted] (12,1)--(13,1);
\draw (14,1) ellipse (0.5cm and 1.75cm);

\draw(1,-1) node{$Q_{r}^{i-2}$};

\draw(4,-1) node{$Q_{r}^{i-1}$};
\draw(6,-1) node{$Q_{r+1}^{i-1}$};
\draw(9,-1) node{$Q_r^i$};
\draw(14,-1) node{$Q_t$};
\draw(2,2.2) node{$v_{i-1}$};
\draw(8,.3) node{$v_{i}$};

\end{tikzpicture}
\label{fig:AlgoPfigure}
\caption{Construction of path $P$}
\end{figure}
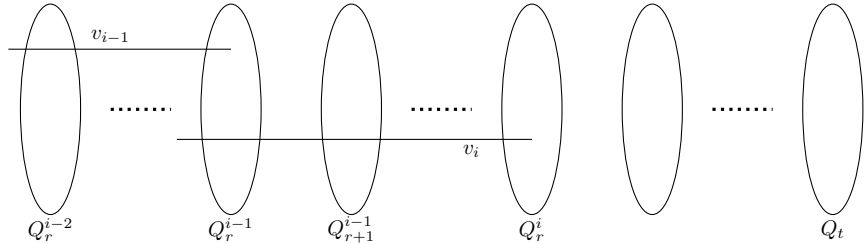

\begin{algorithm}[t]
\caption{Construction of path $P$ and computation of clique cover\\
\textbf{\textit{Input}}: An interval graph $G$ with a linear ordering of maximal cliques $\mathcal{O} = \{Q_1, Q_2 \ldots Q_t\}$\\
\textbf{\textit{Output}}: Path $P$}
\label{algo:constructP}
\begin{algorithmic}[1]
\STATE $i = 1$ 
\STATE $v_1 \gets $ \texttt{Interval in $Q_1$ with largest right endpoint}
\STATE $P \gets v_1$
\STATE $Q_r^1 = $ \texttt{Maximal clique in which $v_1$ ends}
\STATE $C(N[v_1]) = $ \texttt{Minimum clique cover of $N[v_1]$}
\STATE $Q_{r'}^1 = $ \texttt{Maximal clique immediately preceding $Q_r^1$ in $C(N[v_1])$
\STATE $CC(N[v_1]) = C(N[v_1])$}
\WHILE {$Q_r^i\neq Q_t$}
\STATE  $i = i + 1$
\STATE $v_i = $ \texttt{Interval $I\in Q_r^{i-1}\setminus Q_{r'}^{i-1}$ which has largest right endpoint; If there are more than one such vertex, then the one with the smallest left endpoint is chosen}
\STATE $P \gets P \cup v_i$
\STATE $Q_r^i$ = \texttt{Maximal clique in which $v_i$ ends}
\STATE $CC(N[v_1,\dots, v_i]) = CC(N[v_1,\dots,v_{i-1}]) \cup C(Q_{r+1}^{i-1},\dots, Q_r^i)$, \texttt{where} $Q_{r+1}^{i-1}$ \texttt{is the maximal clique immediately succeeding $Q_r^{i-1}$ in $\mathcal{O}$} \label{step:CC}
\STATE $ Q_{r'}^i $ = \texttt{Maximal clique immediately preceding $Q_r^i$ in $CC(N[v_1,\dots, v_i])$}
\ENDWHILE
\STATE $\mathcal{K} = CC(N[v_1,\dots, v_i]) $ \label{step:assign}
\STATE return $P$
\end{algorithmic}
\end{algorithm}

 Let $\{v_1, v_2, \dots, v_p\}$ be the ordered set of vertices in the constructed path $P$  with respect to the linear ordering $\mathcal{O}$. 
 Let $v_i, \dots, v_j, 1\leq i \leq j \leq p$ be  any subset of  vertices  in path $P$. We use $CC(N[v_i,v_{i+1},\ldots,v_j]),~ i \leq j$ to denote a clique cover  of $(N[v_i]\cup N[v_{i+1}]\cup\ldots\cup N[v_j])$ and $|CC(N[v_i,v_{i+1},\ldots,v_j])|$ to denote the number of cliques in $CC(N[v_i,v_{i+1},\ldots,v_j])$.  Note that $CC(N[v_i, \ldots, v_j])$ is a clique cover of graph $G$  when $i=1, j=p$. 
Thus obtained clique cover of $G$, $CC(N[v_1, \ldots, v_p])$, is stored in $\mathcal{K}$. We denote the maximal cliques which constitute $\mathcal{K}$  in the order in which they appear in $CC(N[v_1,\dots, v_p])$,  by  $K_1, K_2,\ldots, K_{\alpha'}$. Here the notation $\alpha'$ is used to indicate that $\mathcal{K}$ is a minimal clique cover of $G$ rather than a minimum clique cover.\\

\noindent  In any perfect graph, the size of a minimum clique cover  equals the size of a maximum independent set. Based on this and the fact that interval graphs are perfect graphs, we state an observation which we use in proving some important properties of the constructed clique  cover of the neighbourhood of vertices in path $P$.
\begin{observation}\label{obs:minCCindSet}
In any perfect graph $G'$, for each maximal clique $K$ in a minimum clique cover $\mathcal{K}$ of $G'$, there exists a vertex $u \in K$ such that $u$ does not belong to any other maximal clique in $\mathcal{K}$.
\end{observation}


\begin{lemma}\label{lem:forbonefour}
For $1 \leq i \leq p$, $| C(N[v_i])|~\leq 3$. 
\end{lemma}
\begin{proof} 
The proof is by contradiction. Let  $\mid C(N[v_i])\mid > 3$. By definition, $C(N[v_i])$  contains only the maximal cliques from the linear ordering $\mathcal{O}$. From Observation \ref{obs:minCCindSet}, it follows that for each maximal  clique $Q \in C(N[v_i])$, there exists a vertex $w$ which is unique to $Q$. Since $| C(N[v_i])|~ > 3$, there exists at least $4$ such vertices each belonging to different maximal cliques in $C(N[v_i])$. Let those vertices be denoted as  $w_1,w_2,w_3,w_4$.  We can easily see that the vertices $w_1,w_2,w_3,w_4$ form an independent set, since each of them belong only to their respective maximal cliques. It follows that $v_i$ together with $w_1,w_2,w_3,w_4$  form a forbidden structure  $K_{1,4}$ (refer Fig. \ref{fig:fbIntGraphs} (i)). This is a contradiction to our premise that $G$ does not contain any forbidden structure. 
 \end{proof}
 
 \begin{lemma}\label{lem:forbtwofive}
In the path $P$, if for any vertex $v_i$, $1\leq i < p$ ,  $| C(N[v_i])|~= 3$, then $| C(N[v_{i+1}])|~\leq 2$. 
\end{lemma}
\begin{proof}
The proof is by contradiction. Assume that there exists a vertex $v_i \in P,~1\leq i \leq p$ for which $|C(N[v_i])| = 3 $ and $| C(N[v_{i+1}])| \geq 3$. Also note that by Lemma \ref{lem:forbonefour}, $| C(N[v_{i+1}])|$ cannot exceed $3$.
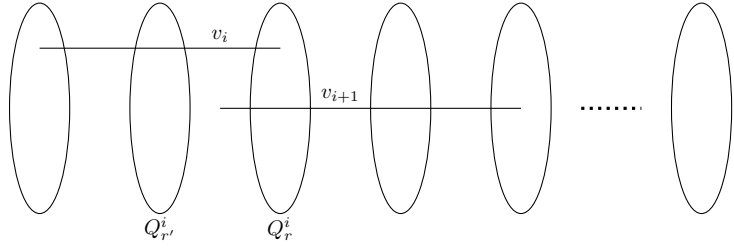
\begin{figure}[h!]
\centering
\begin{tikzpicture}[scale=0.8, every node/.style={scale=0.8}]
\draw(3,2)--(7,2);
\draw(6,1)--(11,1);

\draw (3,1) ellipse (0.5cm and 1.75cm);
\draw (5,1) ellipse (0.5cm and 1.75cm);
\draw (7,1) ellipse (0.5cm and 1.75cm);
\draw (9,1) ellipse (0.5cm and 1.75cm);
\draw (11,1) ellipse (0.5cm and 1.75cm);
\draw[very thick,dotted] (12,1)--(13,1);
\draw (14,1) ellipse (0.5cm and 1.75cm);

\draw(5,-1) node{$Q_{r'}^i$};

\draw(7,-1) node{$Q_r^i$};
\draw(6,2.2) node{$v_i$};
\draw(8,1.2) node{$v_{i+1}$};

\end{tikzpicture}
\label{fig:ProofFig2}
\caption{Forbidden structure formation}
\end{figure}
Vertices $v_i$ and $v_{i+1}$, being adjacent,  form an edge in the  path $P$.  
We consider the following cases  based on the cardinality of $C(N[v_i]) \cup C(N[v_{i+1}])$. \\

\noindent \textbf{Case when $\mid C(N[v_i]) \cup~ C(N[v_{i+1}])\mid = 4$}: Recall from Algorithm \ref{algo:constructP}  that $Q_r^{i-1}$ and $Q_r^i$ are the maximal cliques in the ordering $\mathcal{O}$ which contains the right endpoints of the intervals corresponding to $v_{i-1}$ and $v_i$ respectively. For any vertex $v_i \in P,~  Q_r^{i-1} \in C(N[v_i])$ and $Q_r^i \in C(N[v_i])$.  By our choice of $v_{i+1}$ in the  construction of path $P$, $v_{i+1} \in \{Q_r^i\setminus Q_{r'}^i\}$. Thus $v_{i+1}$  is covered by $Q_r^i$ in $C(N[v_i])$. As per our assumption, $\mid C(N[v_i]) \cup C(N[v_{i+1}])\mid~= 4$, and by our premise, $\mid C(N[v_i])\mid~= 3$.  Therefore those vertices of $N[v_{i+1}]$ which are not covered by $Q_r^i$ have to be covered by exactly one more clique. It follows that $\mid C(N[v_{i+1}])\mid~= 2$, 
which is a contradiction to our assumption that $\mid C(N[v_{i+1}])\mid~=3$. This, in turn, is a contradiction to our initial premise that $\mid C(N[v_i]) \cup~ C(N[v_{i+1}])\mid = 4$. Thus the only possibility is, $\mid C(N[v_i]) \cup C(N[v_{i+1}])\mid = 5$, which we discuss in the next case.  Observe that $\mid C(N[v_i]) \cup C(N[v_{i+1}])\mid$ cannot be greater than 5 since $v_{i+1}\in Q_r^i$ and $ \mid C(N[v_{i+1}])\mid~= 3$.  Therefore, if $\mid C(N[v_i]) \cup C(N[v_{i+1}])\mid$ is greater than 5, then those vertices in $N[v_{i+1}]$ which are not covered in $Q_r^i$  will be covered by 3 maximal cliques, which would make $\mid C(N[v_{i+1}])\mid = 4$. This is again a contradiction to the assumption that $ \mid C(N[v_{i+1}])\mid~= 3$.\\

\noindent \textbf{Case when $\mid C(N[v_i]) \cup C(N[v_{i+1}])\mid = 5$}: The proof is by contradiction to our premise that $G$ does not contain any forbidden structure. We  first show that  $C(N[v_i]) \cup C(N[v_{i+1}])$ is indeed a minimum clique cover of $N[v_i]\cup N[v_{i+1}]$. Then, using Observation \ref{obs:minCCindSet}, we show that there exists a forbidden structure. By definition, $C(N[v_i])$ is a minimum clique cover of $N[v_i]$. Therefore,  each of the three maximal cliques in $C(N[v_i])$ has atleast one  unique vertex which does not belong to any other maximal clique. Since $v_{i+1}\in Q_r^i$ and $Q_r^i\in C(N[v_i])$, $Q_r^i\in C(N[v_{i+1}])$. Let the other two maximal cliques in $C(N[v_{i+1}]$ be $Q_j$ and $Q_k$. By  Observation \ref{obs:minCCindSet}, $Q_j$ and $Q_k$ contain a unique vertex each. It follows that any minimum clique cover of $N[v_i]\cup N[v_{i+1}]$ contains all three maximal cliques of $C(N[v_i])$ along with $Q_j$ and $Q_k$. Hence $C(N[v_i]) \cup C(N[v_{i+1}])$  is a minimum clique cover of  $N[v_i]\cup N[v_{i+1}]$. By  Observation \ref{obs:minCCindSet}, on $C(N[v_i]) \cup~ C(N[v_{i+1}])$, there is a set $V'$ of  $5$ vertices in $C(N[v_i]) \cup~ C(N[v_{i+1}])$ that are mutually disjoint and form an independent set of size five. The edge $(v_i, v_{i+1})$,  together with $V'$ form a forbidden structure (see Definition \ref{def:forbidPattern}). Thus we have arrived at a  contradiction. Therefore, if $\mid C(N[v_i])\mid~= 3$, then $\mid C(N[v_{i+1}])\mid~\leq 2$.
 \end{proof}
 
 \begin{observation}\label{obs:ccatleast2}
 For each vertex $v\in P\setminus v_p$, $\mid C(N[v])\mid \geq 2$.
 \end{observation}
 \begin{proof}
 By construction of path $P$,
 \[ CC(N[v_1,\dots,v_i]) = CC(N[v_1,\dots,v_{i-1}]) \cup C(Q_{r+1}^{i-1},\dots, Q_{r}^i) \]
$Q_r^{i-1}$ is the rightmost maximal clique in $CC(N[v_1,\dots,v_{i-1}])$ and it covers $v_i$, since $v_i \in Q_r^{i-1}$.  $Q_{r}^i$ is the rightmost maximal clique which $v_i$ belongs to, in the ordering $\mathcal{O}$. Note that $C(N[v_i])= Q_r^{i-1}\cup C(Q_{r+1}^{i-1},\dots, Q_{r}^i)$. Let $A = N[v_i]\cap (Q_{r+1}^{i-1} \cup\dots\cup Q_{r}^i)$. Since $v_i\neq v_p$, we know that there is $v_{i+1}\in P$ which is chosen such that $v_{i+1}\notin Q_r^{i-1}$ and $v_{i+1}\in N[v_i]$. It follows that $A\neq  \emptyset$. By choice of $v_i$, it has the rightmost right endpoint among all vertices in $ Q_r^{i-1}\setminus Q_{r'}^{i-1}$ and $v_i\neq v_p$. Hence $\exists u\in A $ which is not covered by $Q_r^{i-1}$. Therefore, there exists at least one $Q \in C(Q_{r+1}^{i-1},\dots, Q_{r}^i)$ that covers all the vertices in $A$,. In other words, $C(Q_{r+1}^{i-1},\dots, Q_{r}^i)\neq \emptyset$. It follows that $C(N[v_i])$ is of size at least 2. 
 \end{proof}
 
\noindent From Lemma \ref{lem:forbtwofive}  and Observation \ref{obs:ccatleast2},  we present the following claim.  
\begin{lemma}\label{claim:forbgeneral}
 In path $P$, there is at most one vertex $v$ where $\mid C(N[v])\mid~= 3$.
\end{lemma}  
\begin{proof} 
The proof of this claim is again by contradiction. Assume that there is more than one vertex with size of minimum clique cover equal to 3 in path $P$. Let $v_i$ be the first such vertex in $P$ in the increasing order of left endpoints. By Lemma \ref{lem:forbtwofive}, we know that the minimum clique cover of $v_{i+1}$ is of size less than 3.  By our assumption, $\exists j > i+1$ such that minimum clique cover of $v_j$ is of size 3. Let the number of vertices in the subpath of $P$ from $v_i$ to $v_j$ (including both $v_i$ and  $v_j$)   be $l$.  It follows from Observation \ref{obs:ccatleast2} that for each vertex $v_k\in P,~ k\in [i+1,j-1]$, the minimum clique cover is of size 2.  We compute $CC(N[v_i, \ldots, v_k])$ with respect to $CC(N[v_i,\dots, v_{k-1}])$. Note that $v_k$ is already covered in $CC(N[v_i,\dots, v_{k-1}])$ and  $CC(N[v_i, \ldots, v_k])$ additionally covers $N[v_k]\setminus N[v_{k-1}]$. Since $|C(N[v_k])|$ = 2, it adds just 1 to the number of cliques in $CC(N[v_i,\dots,$ $ v_{k-1}])$. That is, 
\[\mid CC(N[v_i, \ldots, v_k])\mid~=~\mid CC(N[v_i,\dots, v_{k-1}])\mid +~1 \]
It follows that each of the $l-2$ vertices in $\{v_{i+1}, \ldots, v_{j-1} \}$ increments the size of the clique cover by 1. That is, $\mid CC(N[v_{i+1},\dots , v_{j-1}])\mid = l-2$. Thus 
\begin{align*}
\mid CC(N[v_i, \ldots, v_j])\mid~ &= ~\mid C(N[v_i])\mid +  \mid CC(N[v_{i+1},\dots, v_{j-1}])\mid +~(\mid C(N[v_j])\mid - 1)\\
&=~ 3 + (l-2) + 3 -1\\
&=~ l+3
\end{align*}

 \begin{figure}[t]
\centering
\begin{tikzpicture}[scale=0.7, every node/.style={scale=0.8}]
\draw(3,2)--(6,2);
\draw(12,0.7)--(16,0.7);
\draw(5.3,1.5)--(7.5,1.5);
\draw(6.8,1.3)--(9,1.3);

\draw (3,1) ellipse (0.5cm and 1.75cm);
\draw (4.5,1) ellipse (0.5cm and 1.75cm);
\draw (6,1) ellipse (0.5cm and 1.75cm);
\draw (7.5,1) ellipse (0.5cm and 1.75cm);
\draw (9.2,1) ellipse (0.5cm and 1.75cm);
\draw[very thick,dotted] (10,1)--(11,1);

\draw (12,1) ellipse (0.5cm and 1.75cm);
\draw (14,1) ellipse (0.5cm and 1.75cm);
\draw (16,1) ellipse (0.5cm and 1.75cm);


\draw(5.2,2.2) node{$v_i$};
\draw(13,1) node{$v_j$};
\draw(6.8,1.8) node{$v_{i+1}$};
\draw(8.3,1.6) node{$v_{i+2}$};
\draw [decorate,decoration={brace,amplitude=5pt,mirror,raise=4ex}]
  (7.5,-0.5) -- (12,-0.5) node[midway,yshift=-3em]{$l-2$};

\end{tikzpicture}
\label{fig:ProofClaimfig}
 \caption{Vertices $v_i$ and $v_j$ belong to three consecutive cliques in the clique cover}
\end{figure}
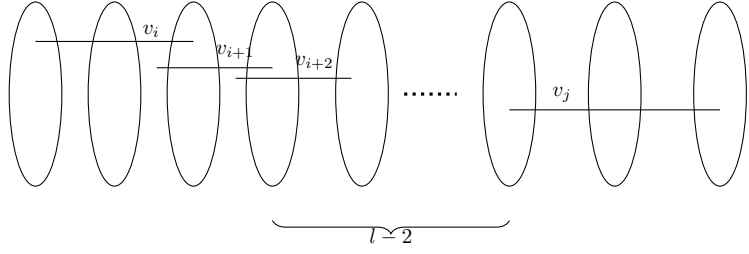
\noindent Note that we deduct 1 from $\mid C(N[v_j])\mid$ since $v_j$ is already covered by $CC(N[v_{i+1}],\ldots,N[v_{j-1}])$. We can see that the vertices from $v_i$ to $v_j$ form a path of length $l$ which has an independent set of size $l+3$ in its neighbourhood. The vertices from  $v_i$ to $v_j$, together with the independent set of size $l+3$ in its neighbourhood forms a forbidden structure. We have arrived at a contradiction to our premise that $G$ does not contain any forbidden structures. Therefore it is proved that in path $P$, there is at most one vertex which has a minimum clique cover of size $3$.
 \end{proof}

\subsection{Computing the exact hitting set from $\mathcal{K}$} \label{subsec:computeExactHittingSet}

We now complete the characterization of \textsf{EHIG}s.  
We first prove the characterization when $p$ is at most 3 and then complete the proof by an inductive argument.  From Algorithm \ref{algo:constructP},  we use the minimal clique cover, $\mathcal{K} = \{K_1, K_2 \dots K_{\alpha'}\}$, which consists of maximal cliques in $\mathcal{O}$  and the path $P$ consisting of vertices $v_1, \ldots, v_p$ from left to right, in the arguments below. Further, for each $1 \leq i \leq \alpha'$, we use $D_i$ to denote the gadget corresponding to $K_i$.  For each $1 \leq i \leq p$, let $L_i$ and $R_i$ denote the leftmost and rightmost, respectively, maximal cliques in $\mathcal{O}$ which contains $v_i$.   $DL_i$ and $DR_i$ denote the gadgets in $H_G$ corresponding to $L_i$ and $R_i$, respectively. Recall that $H_G$ denotes the canonical interval representation of $G$. It is useful to note that $DL_1$ and $D_1$ both denote the gadget associated with $K_1$, and $DR_p$ and $D_{\alpha'}$ both denote the gadget associated with $K_{\alpha'}$.  The inductive argument below refers to two interval graphs $G$ and $G'$, and the gadgets we refer to are in $H_G$ or $H_{G'}$, and will be clear from the context.  In all our arguments that follow, we use the maximal clique ordering $\mathcal{O}$ to reason about the position of a maximal clique with respect to another maximal clique.  We thus use the words and phrases {\em before, after, at or before, at or after} with respect to $\mathcal{O}$.  These relationships are also represented by $<, >, \leq, \geq$, whenever it is more convenient to use these symbols.  

We now prove the characterization based on the different cases for $p$. For $p=1$, we explicitly present the hitting set, and when $2 \leq p \leq 3$, we   present the exact hitting sets which satisfy some additional properties. These additional properties are useful in the inductive argument for $p \geq 3$.  In all the following lemmas in this section, we assume that $G$ is an interval graph which does not have an $F \in {\cal F}$ as an induced subgraph.  Further, all the statements are based on the value of $p$ which is the number of vertices in the path $P$ computed by Algorithm \ref{algo:constructP}, and ${\cal K}$.  
\begin{lemma}
    \label{lem:basecase1}
Let $v_1$ be the only vertex in $P$.   Then the canonical interval representation $H_G$ satisfies one of the following two statements.
 \begin{itemize}
     \item Let $N[v_1]$ have a minimum clique cover of size 2, which is $\{K_1, K_2\}$. Then $\{z_1, z_2+1\}$ and $\{z_1-1,z_2\}$ are two exact hitting sets of $H_G$.
     \item Let $N[v_1]$ have a minimum clique cover of size 3, which is $\{K_1, K_2, K_3\}$.  Let $w_1 = |K_1 \cap K_2|$  and let $w_3 = |K_2 \cap K_3|$.  Then $\{z_1 - w_1, z_2, z_3 + w_3\}$ is an exact hitting set of $H_G$.
 \end{itemize}
\end{lemma}
\begin{proof}
In the first case, $K_1$ and $K_2$ are the first and last cliques of $\mathcal{O}$, respectively.  Let $D_1$ and $D_2$ be the gadgets corresponding to $K_1$ and $K_2$, respectively, in $G$.  By the definition of $H_G$, the interval associated with $v_1$ has $z_1$ as the left endpoint and $z_2$ as the right endpoint. Consider the set $\{z_1, z_2 + 1\}$. $z_1$ hits all the intervals whose left endpoint is in $D_1$, and $z_2+1$ hits all the intervals whose left endpoint is to the right of $D_1$, and right endpoint is to the right of $z_2$ in $D_2$.  Further, each interval is hit, and no interval is hit twice. By a symmetric argument, $\{z_1-1,z_2\}$ is also an exact hitting set.

 In the second case, consider the cliques $B_1 = K_1 \setminus K_2$, $B_3 = K_3 \setminus K_2$, and $B_2 = K_2$.  Let $w_1 = |K_1 \cap K_2|$ and $w_3 = |K_3 \cap K_2|$. We show that the set $\{z_1 - w_1, z_2, z_3 + w_3\}$ is an exact hitting set of $H_G$.  First, we observe that the point $z_2$ hits all the intervals in $D_2$. Further, from the construction of $H_G$ it follows that, if $I_1$ is an interval in $D_1 \setminus D_2$ and $I_2$ is an interval in $D_1 \cap D_2$, then in $D_1$ the left endpoint of $I_1$ is smaller than the left endpoint of $I_2$.   Thus, among the $w_1$ intervals in $D_1 \cap D_2$, the longest interval has the left endpoint in $z_1$ and the remaining $w_1 - 1$ intervals start at different points in $\{(z_1-1),\ldots,(z_1-(w_1-1))\}$. Therefore, the point $z_1 - w_1$  does not hit any of the $w_1$ intervals that belong to  $D_1 \cap D_2$. Further, it hits all the intervals in $D_1 \setminus D_2$.  A symmetric argument shows that $z_3 + w_3$ hits all the intervals in $D_3 \setminus D_2$ and does not hit any interval in $D_3 \cap D_2$.  The remaining intervals are all in $D_2$, and they are not hit by either $z_1 - w_1$ or $z_3 + w_3$, and they all contain $z_2$ which is the zero-point of gadget $D_2$.  
 Thus, this base case is proved.  
\end{proof}
\begin{lemma}
    \label{lem:basecase2}
    Let $P$ consist of only $v_1$ and $v_2$. Then, the following statements hold for $H_G$.
    \begin{itemize}
        \item Let $N[v_1]$ and $N[v_2]$ have minimum clique cover of size 2. In ${\cal K} = \{K_1,K_2,K_3\}$, let $w_1 = |K_1 \cap L_2|$ and let $w_2 = |K_{3} \cap K_{2}|$.  Then $\{z_1 - w_1, z_2, z_3 + w_3\}$ is an exact hitting set of $H_G$.
        \item Let $N[v_1]$ have minimum clique cover of size 2 and $N[v_2]$ have minimum clique cover of size 3. In ${\cal K} = \{K_1,K_2, K_3, K_4\}$, let $w_1 = |K_1 \cap L_2|$ and $w_4 = |K_4 \cap K_3|$. Then $H_G$ has an exact hitting set which contains $\{z_1-w_1, z_4+w_4\}$.
        \item Let $N[v_1]$ have minimum clique cover of size 3 and $N[v_2]$ have minimum clique cover of size 2. In ${\cal K} = \{K_1,K_2, K_3, K_4\}$, let $w_1 = |K_1 \cap K_2|$ and $w_4 = |K_4 \cap K_3|$. Then $H_G$ has an exact hitting set which contains $\{z_1-w_1, z_4+w_4\}$.
    \end{itemize}
\end{lemma}
\begin{proof}
If $p=2$, then ${\cal K}$ either has 3 maximal cliques or 4 maximal cliques. Thus $\alpha'$ which is the index of the last clique in ${\cal K}$ is in $\{3,4\}$.  Further, in all the cases, $z_1 - w_1$ is in $D_1$ and for $\alpha' \in \{3,4\}$, $z_{\alpha'} + w_{\alpha'}$ is in $D_{\alpha'}$.   Consider the point $z_1 - w_1$ which is in $D_1$.  In the first two cases, $z_1 - w_1$ hits the intervals associated with vertices in $K_1 \setminus L_2$ and does not hit any interval associated with vertices in  $K_1 \cap L_2$.  Consider the interval graph obtained by removing $K_1 \setminus L_2$.  Since $N[v_1]$ has a clique cover of size 2, it follows that $L_2 \setminus K_1 \subset K_2$.   
Consider the maximal cliques from $K_2$ to $K_3$.  $v_2$ is common to all of these cliques.  We now consider the first case and the second case separately.

In the first case,  the minimum clique cover of $N[v_2]$ is of size 2.   Let $w_3$ be the number of vertices in $K_2 \cap K_3$.  Consider the exact hitting set $\{z_2, z_3+w_3\}$. Clearly, $z_2$ hits all the intervals in $D_2$, and $z_3 + w_3$ hits all the intervals in $D_3 \setminus D_2$ and does not hit anything in $D_2$; this is because the right endpoints of the intervals in $D_2 \cap D_3$ are smaller than $z_3 + w_3$.  Thus, for the first case $\{z_1-w_1, z_2, z_3+w_3\}$ is an exact hitting set of $H_G$. 

In the second case, the minimum clique cover of $N[v_2]$ is of size 3.  Since the first clique containing $v_2$ is $L_2$, the minimum clique cover of $N[v_2]$ is $L_2, K_3, K_4$. Consider $G'$ to be the interval graph induced by $N[v_2]$ and consider $H_{G'}$.  $L_2$ is the first maximal clique in $G'$ and in this case let $z$ denote the zero-point of the gadget of $DL_2$ in $H_G$.
By using the second case of Lemma \ref{lem:basecase1}, consider the exact hitting set $\{z - |L_2 \cap K_3|, z_3, z_4 + |K_3 \cap K_4|\}$ of $H_{G'}$.    From this set, we construct an exact hitting set of $H_G$ based on the following two subcases.  The first subcase is that there are vertices in $L_2 \setminus K_1$ whose leftmost clique is $L_2$, and whose rightmost clique  is before $K_3$.  Then, the interval in $H_{G'}$ of such a vertex is hit by $z - |L_2 \cap K_3|$ and this also hits all the intervals associated with vertices in $L_2 \setminus K_3$.  Thus, in this subcase, we get an exact hitting set for $H_G$ consisting of $\{z_1-w_1, z - |L_2 \cap K_3|, z_3, z_4 + |K_3 \cap K_4|\}$.  In the second subcase, for all the vertices whose leftmost clique is $L_2$, the rightmost clique is at or after $K_3$.  Let $h$ be the intermediate point just preceding $DL_2$, which is the gadget associated with $L_2$ in $H_G$.  
Consider the set $\{z_1 - w_1, h, z_3, z_4 + |K_3 \cap K_4|\}$.  $z_1 - w_1$ hits all intervals associated with vertices in $K_1 \setminus L_2$, $h$ hits all intervals associated with vertices in $L_2 \setminus (L_2 \cap K_3)$, and this set includes $K_1 \cap L_2$.  $z_3$ hits all intervals associated with vertices in $(L_2 \cap K_3) \cup (K_3 \cap K_4)$, and $z_4 + |K_3 \cap K_4|$ hits all intervals associated with vertices in $K_4 \setminus K_3$ (that is, the intervals in $D_4 \setminus D_3$).   Further, it is clear that it is an exact hitting set. Thus, the first two cases are proved.  
The third case is symmetric to the second case\footnote{A detailed description of the symmetry is presented in the inductive argument in Lemma \ref{lem:charcomplete}}, and thus it is also true. Hence the lemma. 
\end{proof}

\noindent
{\em Remark:} The approach of obtaining an exact hitting set for $H_G$ from an exact hitting set for $H_{G'}$, by using a preceding or following intermediate point, is a template that repeats in the following proofs. However, this template is used in different contexts, and it seems that the repeated case-by-case usage of this template is unavoidable.
\begin{lemma}
    \label{lem:basecase3}
    Let $P$ consist of only $v_1, v_2, v_3$, and let $N[v_2]$ have a minimum clique cover of size 3.  Further, let $w_1 = |K_1 \cap L_2|$ and let $w_{\alpha'} = |K_{\alpha'} \cap K_{\alpha'-1}|$.  Then $H_G$ has an exact hitting set containing $\{z_1 - w_1, z_{\alpha'} + w_{\alpha'}\}$.
\end{lemma}
\begin{proof}
    Consider the interval graph $G'$ obtained by removing $K_1 \setminus L_2$.  Let $H_{G'}$ be the canonical  interval representation.
    The path obtained from Algorithm \ref{algo:constructP} on $G'$, using the cliques from $L_2$ to $K_{\alpha'}$ is the path consisting of $v_2$ and $v_3$ only.  The first maximal clique in $H_{G'}$ is $L_2$ which is the leftmost clique containing $v_2$. Further, the rightmost clique containing $v_2$ is $K_{\alpha'-1}$, and the minimum clique cover of $N[v_2]$ is 3 and the minimum clique cover of $N[v_3]$ is 2 (from Lemma \ref{claim:forbgeneral}).  Let $S$ be an exact hitting set of $H_{G'}$ obtained from using Lemma \ref{lem:basecase2}.  Let $z$ denote the zero-point of the gadget $DL_2$.  For some $w > 0$, let $z-w$ be in $S$.  Also,  in the construction of $S$ it is ensured that the interval associated with $v_2$ is not hit by $z-w$.  From the structure of $S$ defined in Lemma \ref{lem:basecase2}, $\{z-w, z_{\alpha'-1}, z_{\alpha'} + w_{\alpha'}\} \subseteq S$.
 We get an exact hitting set for $H_G$ from $S$ based on two cases.
    
    The first case is that there are vertices in $L_2 \setminus K_1$ for which the leftmost clique is $L_2$ and the interval associated with them are hit by $z-w$ in $H_{G'}$.   Such vertices will also occur in the rightmost clique containing $v_1$, since $N[v_1]$ has a clique cover of size 2 in $G$.  Thus in this case, $S \cup \{z_1 - w_1\}$ is an exact hitting set for $H_G$.    The main  reason  is that the intervals associated with the vertices in $L_2 \setminus K_1$ whose leftmost clique is at  or before $L_2$ will all be hit by $z-w$. Further, $z_1 - w_1$ hits all the intervals associated with vertices in $K_1 \setminus L_2$.
    
    In the second case, all the intervals in $L_2 \setminus K_1$ whose leftmost clique is $L_2$ are not hit by $z-w$; they are hit by another element of $S$.  Thus, $z-w$ is in the hitting set of $H_{G'}$ only to hit intervals whose left endpoint is to the left of $DL_2$ in $H_G$. Let $h$ be the intermediate point just to preceding the gadget associated with $DL_2$ in $H_G$.  
Consider the set $S \setminus \{z-w\} \cup \{z_1 - w_1, h\}$.  This is an exact hitting set of $H_G$ since $S$ is an exact hitting set of $H_{G'}$, the intervals associated with the vertices in $L_2 \setminus K_1$ whose leftmost clique is before $L_2$ will all be hit by $h$, and $z_1 - w_1$ hits all the intervals associated with vertices in $K_1 \setminus L_2$. Hence the lemma.
 \end{proof}
 
\begin{lemma}
\label{lem:charcomplete}
Let $G$ be an  interval graph  which  does not have an $F \in {\cal F}$ as an induced subgraph.  Then the canonical interval representation $H_{G}$ has an exact hitting set. Further, if $p \geq 2$, there is an exact hitting satisfying the following two additional properties: 
\begin{itemize}
    \item The intervals corresponding to the vertices in $K_1 \setminus N[v_2]$ are hit by a point in $D_1$ to the left of $z_1$ and $D_1 \cap N[v_2]$ is hit by a point in $DL_2$ which is to the left of the zero-point of $DL_2$ or the intermediate point just before $DL_2$.
    \item The intervals corresponding to vertices  in $K_{\alpha'} \setminus N[v_{p-1}]$ are hit by a point in $D_{\alpha'}$ to the right of $z_{\alpha'}$ and $D_{\alpha'} \cap N[v_{p-1}]$ is hit by a point in $D_{\alpha'-1}$  to the right of the zero-point of $D_{\alpha'-1}$ or the intermediate point just after  $D_{\alpha'-1}$.
\end{itemize}
\end{lemma}
\begin{proof}
The proof when $p=1$ follows from Lemma \ref{lem:basecase1}.
The proof is by induction on $p \geq 2$.  The base case is for $p=2$, and from Lemma \ref{lem:basecase2} we know $H_G$ has a hitting set which satisfies the additional properties.  
We now prove the claim for all $p \geq 3$. 

First, we consider the case in which  minimum clique cover of  either $N[v_1]$ or $N[v_2]$ is of size 3 and  the minimum clique cover of either $N[v_{p-1}]$ or $N[v_p]$ is of size 3.  From Lemma \ref{claim:forbgeneral}, we know that there is at most one vertex in $P$ which has a minimum clique cover of size 3.
Therefore, in this case, it follows that $p=3$, and $N[v_2]$ has a minimum clique cover of size 3.  From Lemma \ref{lem:basecase3}, we know that $H_G$ has an exact hitting set which satisfies the additional properties.  

Next we consider the case in which  either the minimum clique cover of  $N[v_1]$ and $N[v_2]$ are both of size 2 and if not, the minimum clique cover of $N[v_{p-1}]$ and $N[v_p]$ are both of size 2.  The induction hypothesis is that the claim is true for all natural numbers  in the set $\{2, \ldots, p-1\}$.  Given this, we prove that the claim is true for $p$.

Let us first consider the case when $N[v_1]$ and $N[v_2]$ both have  a minimum clique cover of size 2.  By construction, in Algorithm \ref{algo:constructP}, we know that $K_2$ is the rightmost clique in $\mathcal{O}$ which contains $v_1$.  From Observation \ref{obs:lastclique}, we know that $K_2$ is in a minimum clique cover of $N[v_1]$. Any vertex in $N[v_1]$ which is not present in $K_2$ will be present in $K_1$, since minimum clique cover of $N[v_1]$ is 2.  Let $L_2$ be the leftmost clique in $\mathcal{O}$ which contains $v_2$.  We know that $K_1 < L_2 \leq K_2$. That is, $L_2$ is a clique to the right of $K_1$ and is not later than $K_2$ in $\mathcal{O}$.
 
Consider the partition of $K_1$ into $K_1 \setminus L_2$ and $K_1 \cap L_2$.  Let $w_1$ be the number of vertices in $K_1 \cap L_2$.  Let $B_1 = K_1 \setminus L_2 = K_1 \setminus N[v_2]$.  Consider the point $z_1 - w_1$ in the gadget $D_1$  in $H_G$.   By the nature of the construction of $H_G$, this point is to the left of the starting point of the intervals associated with vertices in $K_1 \cap L_2$.  Consider the interval graph $G'$ obtained by removing $B_1$ from $G$.  Since the elements in $N[v_1]$ are in the cliques $K_1$ and $K_2$, the first clique in the maximal clique ordering of $G'$ is $L_2$.  Further, among all the vertices in $L_2$, $v_2$ is the vertex in $L_2$ for which the rightmost clique containing it has the largest index in $\mathcal{O}$. Thus, Algorithm \ref{algo:constructP} on $G'$ will compute the path $P \setminus v_1$ which has the $p-1$ vertices $v_2, v_3, \ldots v_p$.  Let $L_3$ denote the leftmost maximal clique in $G'$ which contains $v_3$.   
The additional properties satisfied by $G'$ are as follows:
\begin{itemize}
    \item For each vertex $v$ in $G'$ for which the leftmost clique containing it (in $\mathcal{O}$) is after $K_1$ and before $L_2$, the rightmost clique containing $v$ is before $K_3$.  Otherwise, the choice for $v_2$ by the algorithm is violated.
    \item For each vertex $v$ in $K_1 \cap L_2$, the rightmost clique (in $\mathcal{O}$) containing $v$ is at or before $K_2$.
    \item For each vertex $u$ in $N[v_2]$ for which the leftmost clique containing it is to the right of $L_2$, $u$ is also present in $K_3$. Otherwise, let 
    there exist such a vertex $u$ for which the rightmost clique containing it is before $K_3$.
    Then $u$ along with one vertex each in $L_2$ and $K_3$ form an independent set of size 3.  Thus,  minimum clique cover of $N[v_2]$ is 3. This contradicts our premise that for both $N[v_1]$ and $N[v_2]$ have a minimum clique cover of size 2.
\item $L_3$ is  after  $K_2$ and not later to $K_3$ in $\mathcal{O}$. 
\end{itemize}
By the induction hypothesis applied to $G'$, $H_{G'}$ has an exact hitting set such that the intervals corresponding to the vertices in $L_2 \setminus N[v_3]$ are hit at a point to the left of the zero-point of the gadget $DL_2$ in $H_{G'}$ and the intervals corresponding to vertices in $L_2 \cap N[v_3]$ are  hit by a point to the left of the zero-point of the gadget  $DL_3$ or by the intermediate point just before $DL_3$.  In particular, $v_2$ which is  in $L_2 \cap N[v_3]$ is hit by a point to the left of the zero-point in the gadget $DL_3$ or by the intermediate point just before $DL_3$.  Note that the intermediate point just before $DL_3$ exists since $DL_3$ is to the right of $D_2$.  Let this hitting set be denoted by $S$.  Let $h^* \in S$ be the point in the gadget $DL_2$ in $H_{G'}$.  Note that $DL_2$ is the leftmost gadget in $H_{G'}$. 

We first show that $h^*$ hits the intervals in $H_{G'}$ corresponding to the vertices in $K_1 \cap L_2$.  This is because, by the construction of $P$, for each vertex in $K_1$, the rightmost clique containing it is at or before $K_2$, and $L_3$ is after $K_2$.  Thus, the point in $H_{G'}$ that hits any interval corresponding to a vertex in $K_1 \cap L_2$ has to be in a gadget to the left of $D_2$.  If this point is different from $h^*$, then it would also hit $v_2$.   This contradicts the fact that in the exact hitting set $S$, $v_2$ is hit by a point to the left of the zero-point in the gadget $DL_3$ or to the intermediate point just before $DL_3$. Clearly, both these points are different from $h^*$ which is a point in the gadget $DL_2$ and we know that $DL_2$ is different from $DL_3$.   

To construct the exact hitting set in $H_G$ from $S$, we define a point $h$ in $H_G$ as follows based on two cases:
\begin{itemize}
    \item The first case is when there is a vertex  such that the leftmost clique containing it is $L_2$ and the rightmost clique containing it is before $K_3$ (in $\mathcal{O})$. Among all such vertices, let $u$ be the vertex such that the rightmost clique containing it has the largest index in $\mathcal{O}$. By construction of $H_{G}$ the left endpoint of the interval associated with $u$ would have been the largest among all such vertices.  We take $h$ to be the left endpoint of the interval assigned to $u$ in $H_G$.
    \item In the second case there is no vertex such that the leftmost clique containing it is $L_2$ and the rightmost clique containing it is before $K_3$ (in $\mathcal{O})$. In this case, $h$ is the intermediate point between $DL_2$ and the preceding gadget in $H_G$.  Such a point exists, since $L_2$ is different from the first clique $K_1$, and by construction of $H_G$, the gadget for every clique in $\mathcal{O}$ except $K_1$ has a point to its left.
\end{itemize} 
Then $S \setminus \{h^*\}  \cup \{z_1 - w_1, h\}$ is a  hitting set of $H_G$, where $z_1 - w_1$ hits all the intervals in $B_1$.  $h$ hits all the intervals associated with vertices in $L_2$ whose rightmost endpoint is before $K_3$, and these are not hit by any other point in $S$, due to the induction hypothesis, and they are not hit by $z_1 - w_1$.  Further, the vertices in $B_1$ are not elements of $L_2$, and thus are hit only by $z_1 - w_1$.  Finally, the intervals associated with all the other vertices are hit exactly once by $S$ in gadgets different from the gadgets associated with $L_2$ and $K_1$, and these gadgets have the same structure in $H_{G'}$ as in $H_G$.  Therefore, $S \setminus \{h^*\}  \cup \{z_1 - w_1, h\}$ is an exact hitting set of $H_G$, and it also satisfies the two additional properties.

Next, we consider the case when $N[v_1]$ or $N[v_2]$ has a minimum clique cover of size 3. Then, since $p \geq 4$, we know that both $N[v_{p-1}]$ and $N[v_p]$ has a minimum clique cover of two cliques. Further, $v_{p-1}$ and $v_p$ are distinct from $v_1$ and $v_2$. From Observation \ref{obs:lastclique}, we know that $K_{\alpha'}$ is in the minimum clique cover of $N[v_p]$.  Consider $v_{p-1}$ in $P$ and we know that $D_{\alpha'-1}$ denotes the rightmost clique in $\mathcal{O}$ which contains $v_{p-1}$.  To prove the claim in this case, we consider the reversed maximal clique ordering $\mathcal{O}$, and the path obtained by executing Algorithm \ref{algo:constructP} on this ordering.  Due to the symmetry of the vertices chosen during the algorithm to be added to the path, it is clear that the path computed will be the reversal of $P$.  Also, we know that $L_{p-1} \leq K_{\alpha'-2} < L_p \leq K_{\alpha'-1} < K_{\alpha'}$.  We consider the argument for the previous case by using the reversal of $\mathcal{O}$, the canonical representation constructed using the reversal of $\mathcal{O}$, and the reversal of $P$. We compare the execution of Algorithm \ref{algo:constructP}  using the reversal of $\mathcal{O}$ and its execution using $\mathcal{O}$:  $K_{\alpha'}$ will be in the place of $K_1$, $K_{\alpha'-1}$  in the place of $L_2$, $L_p$  in the place of $K_2$, $K_{\alpha'-2}$ in the place of $L_3$, and $L_{p-1}$ in  place of $K_3$.  
 By an argument symmetric with the argument using the induction hypothesis for the previous case, using the canonical representation constructed using the reversal of $\mathcal{O}$, it follows that $S \setminus \{h^*\} \cup \{z_{\alpha'} + w_{\alpha'},h\}$ is an exact hitting set of $H_G$, and it satisfies the two additional properties. 
Hence the lemma.         
\end{proof} 

\noindent
We complete the proof of the  forbidden structure characterization for \textsf{EHIG}s. 
\forbidden*
\begin{proof}
From Section \ref{sec:canonRep}, we know that every interval graph $G$ has a unique canonical interval representation, which we denote by $H_G$. Furthermore, if $G$ does not have an $F \in {\cal F}$ as an induced subgraph, then by Lemma \ref{lem:charcomplete}, $H_G$ has an exact hitting set. 
    We have shown in Lemma \ref{lem:necFbIntGraph} that if $G$ has an exactly hittable interval representation, then $G$ does not have any $F \in {\cal F}$ as an induced subgraph. This proves the theorem.
 \end{proof}
 
\noindent
Using the above theorem, we prove Theorem \ref{thm:EHIHeqEHIG}. 
\equivalence*
\begin{proof}
Let the hypergraph $H_G$ constructed from interval graph $G$ has an exact hitting set. By Lemma \ref{lem:canonEqIntGraph}, the intersection graph $G'$ of $H_G$ is isomorphic to $G$. It follows that if $G'$ has an exactly hittable interval representation, then $G$ also has one. Thus, $G$ is exactly hittable.

To show the other direction, let $G$ be an $\ehig$. By Theorem \ref{thm:EHInterGraphs}, if $G$ is exactly hittable, then $G$ does not have any forbidden structure. Then, it follows from Lemma \ref{lem:charcomplete} that the canonical representation $H_G$ of $G$ has an exact hitting set.  
 \end{proof}

\subsection{Algorithm to recognize exactly hittable interval graphs}
\label{sec:algoEHIG}
In this section, we present an algorithm to recognize an exactly hittable interval graph. This algorithm makes use of the canonical interval representation in Section \ref{sec:canonRep} and the result by \cite{Dom2006} for MMSC problem (described in Section \ref{sec:GenFramework}) on interval hypergraphs. In their paper, Dom et al. showed that an integer linear programming (ILP) formulation, say $\mathcal{L}$, for MMSC problem on interval hypergraphs can be solved in polynomial time. The coefficients of inequalities in $\mathcal{L}$ results in a totally unimodular matrix and the polyhedron corresponding to $\mathcal{L}$ is an integer polyhedron. If the input instance to ILP is an exactly hittable instance, then the solution returned is 1. We use this algorithm below to test if a given interval hypergraph instance is exactly hittable.  \\

\noindent
\textbf{Algorithm} \texttt{isEHIG}: Given an interval graph $G$, construct the canonical interval representation as described in Section \ref{sec:canonRep}. Let $H_G$ be the resulting interval representation. Run MMSC algorithm by  \cite{Dom2006} on $H_G$ as input. If the algorithm returns value 1, then return \texttt{yes}. Else return \texttt{no}.

\begin{lemma}
Algorithm \texttt{isEHIG}($G$) outputs yes if and only if $G$ is exactly hittable in polynomial time.
\end{lemma}
\begin{proof}
The proof follows from Lemma \ref{lem:canonEqIntGraph}, Theorem \ref{thm:EHIHeqEHIG} and the correctness of algorithm for MMSC problem on interval hypergraphs.
\end{proof}

\noindent
It is also clear that the inductive argument in Lemma \ref{lem:charcomplete} can be converted into a polynomial time combinatorial algorithm to check if $H_G$ has an exact hitting set. This leverages the fact that a minimum clique cover of a perfect graph can be computed in polynomial time.

\subsection{Proper Interval Graphs is a subclass of \textsf{EHIG}}
\label{subsec:PropSubSetEHIG}
We now recall and complete the proof of Theorem \ref{thm:Hierarchy}.  
\hierarchy*
\begin{proof}
Let $G$ be a proper interval graph and let it be the intersection graph of the interval hypergraph $H = (\mathcal{V},\mathcal{I})$  in which no interval properly contains another. Since $H$ is a proper interval hypergraph, no two intervals in $\mathcal{I}$ can have the same left endpoint. Hence order intervals in $\mathcal{I}$ according to increasing order of their left endpoints. Let this ordering be $I_1 < I_2 < \ldots < I_m$. Add $r(I_1)$ (which is the smallest right endpoint among all intervals) to set $S$. Remove all intervals hit by $r(I_1)$. Recurse on the remaining set of intervals until all the intervals are hit by $S$. Clearly, $S$ is an exact hitting set.

To show the strict containment, we show that the graph $K_{1,3}$  which is a forbidden structure (\cite{roberts1978graph}) for Proper Interval Graphs has an exactly hittable interval representation.  
Let the vertices of the $K_{1,3}$ be $\{u,a,b,c\}$ and edges be $\{(u,a),(u,b),(u,c)\}$. The intervals assigned to the vertices $a,b,c$ and $u$ are shown in Fig. \ref{fig:K13}. Hence the lemma.
\end{proof}

\begin{figure}
\centering
\begin{tikzpicture}[scale=0.8, every node/.style={scale=0.8}]
\draw(0,0)--(1,0.5);
\draw(0,0)--(1,0);
\draw(0,0)--(1,-0.5);
\draw(-0.2,0) node{$u$}; 
\draw(1.2,0.5) node{$a$}; 
\draw(1.2,0) node{$b$}; 
\draw(1.2,-0.5) node{$c$}; 
 
\draw[radius=0.3mm, color=black, fill=black](0,0) circle; 
\draw[radius=0.3mm, color=black, fill=black](1,0.5) circle;
\draw[radius=0.3mm, color=black, fill=black](1,0) circle; 
\draw[radius=0.3mm, color=black, fill=black](1,-0.5) circle; 

\draw(3,0)--(4,0);
\draw(3.5,0.2) node{$a$}; 
\draw[radius=0.3mm, color=black, fill=black](3,0) circle; 
\draw[radius=0.3mm, color=black, fill=black](4,0) circle; 
\draw[radius=0.3mm, color=black, fill=black](5,0) circle node[above=0.02cm]{$b$};
\draw(6,0)--(7,0) (6.5,0.2) node{$c$};  
\draw[radius=0.3mm, color=black, fill=black](6,0) circle; 
\draw[radius=0.3mm, color=black, fill=black](7,0) circle; 
\draw(4,1)--(6,1) (5,1.2) node{$u$};  
\draw[radius=0.3mm, color=black, fill=black](4,1) circle; 
\draw[radius=0.3mm, color=black, fill=black](6,1) circle; 

\draw[radius=0.3mm, color=black, fill=black](3,-0.5) circle node[below=0.02cm] {$\textcolor{black}{1}$};
\draw (3,-0.5) -- (7,-0.5);
\draw[radius=0.3mm, color=black, fill=black](4,-0.5) circle node[below=0.02cm] {$\textcolor{black}{2}$};
\draw[radius=0.3mm, color=black, fill=black](5,-0.5) circle node[below=0.02cm] {$\textcolor{black}{3}$};
\draw[radius=0.3mm, color=black, fill=black](6,-0.5) circle node[below=0.02cm] {$\textcolor{black}{4}$};
\draw[radius=0.3mm, color=black, fill=black](7,-0.5) circle node[below=0.02cm] {$\textcolor{black}{5}$};
\end{tikzpicture}
\caption{Exactly hittable interval representation of $K_{1,3}$ (here, \{1,3,5\} is an exact hitting set)}
\label{fig:K13}
\end{figure}
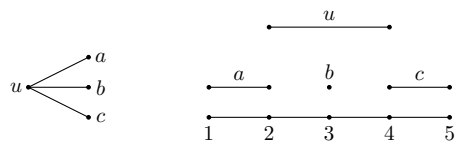

\noindent
\section{Discussion}
Our results indicate that there is an interesting hierarchy among the class of interval graphs based on the number of times an interval is hit by a hitting set.  We have shown that proper interval graphs have an exactly hittable interval representation. Further it is a strict subclass of the set of $\ehig$s. The natural question is to characterize interval graphs which have a representation such that each interval is hit at most $k$ times by a hitting set.  We also believe that the recognition problem for such graphs is fundamental and interesting.

\acknowledgements{
We thank the anonymous reviewers for providing insightful comments on one of the gaps in the crucial proof of characterization. We would also like to thank the organizers of ICGT 2022 for their support during the review process. }

\nocite{*}
\bibliographystyle{abbrvnat}
\bibliography{cfc}

\begin{thebibliography}{71}
\providecommand{\natexlab}[1]{#1}
\providecommand{\url}[1]{\texttt{#1}}
\expandafter\ifx\csname urlstyle\endcsname\relax
  \providecommand{\doi}[1]{doi: #1}\else
  \providecommand{\doi}{doi: \begingroup \urlstyle{rm}\Url}\fi

\bibitem[Abel et~al.(2017)Abel, Alvarez, Demaine, Fekete, Gour, Hesterberg,
  Keldenich, and Scheffer]{Abel2017}
Z.~Abel, V.~Alvarez, E.~D. Demaine, S.~P. Fekete, A.~Gour, A.~Hesterberg,
  P.~Keldenich, and C.~Scheffer.
\newblock Three colors suffice: Conflict-free coloring of planar graphs.
\newblock In \emph{SODA}, 2017.

\bibitem[Alon and Spencer(2015)]{ANJS2015}
N.~Alon and J.~H. Spencer.
\newblock \emph{The probabilistic method}.
\newblock John Wiley \& Sons, 2015.

\bibitem[Ashok et~al.(2015)Ashok, Dudeja, and Kolay]{MAXCFC2015}
P.~Ashok, A.~Dudeja, and S.~Kolay.
\newblock \emph{Exact and FPT Algorithms for Max-Conflict Free Coloring in
  Hypergraphs}, pages 271--282.
\newblock Springer Berlin Heidelberg, Berlin, Heidelberg, 2015.
\newblock ISBN 978-3-662-48971-0.
\newblock \doi{10.1007/978-3-662-48971-0_24}.
\newblock URL \url{http://dx.doi.org/10.1007/978-3-662-48971-0_24}.

\bibitem[Bar-Noy et~al.(2006)Bar-Noy, Cheilaris, and Smorodinsky]{ABPCSS2006}
A.~Bar-Noy, P.~Cheilaris, and S.~Smorodinsky.
\newblock Conflict-free coloring for intervals: From offline to online.
\newblock In \emph{Proceedings of the Eighteenth Annual ACM Symposium on
  Parallelism in Algorithms and Architectures}, SPAA '06, pages 128--137, New
  York, NY, USA, 2006. ACM.
\newblock ISBN 1-59593-452-9.
\newblock \doi{10.1145/1148109.1148133}.
\newblock URL \url{http://doi.acm.org/10.1145/1148109.1148133}.

\bibitem[Bhattacharya and Chalermsook(2014)]{Sayan2014}
S.~Bhattacharya and P.~Chalermsook.
\newblock Approximation algorithms.
\newblock \emph{Topics in Approximation Algorithms (Winter 2013/14), Max Planck
  Institute}, 2014.
\newblock \url{http://resources.mpi-inf.mpg.de/departments/
  d1/teaching/ws13/Approx/sol3.pdf}.

\bibitem[Bhyravarapu et~al.(2021)Bhyravarapu, Hartmann, Kalyanasundaram, and
  Vinod~Reddy]{Bhyravarapu2021}
S.~Bhyravarapu, T.~A. Hartmann, S.~Kalyanasundaram, and I.~Vinod~Reddy.
\newblock Conflict-free coloring: Graphs of bounded clique width and
  intersection graphs.
\newblock In P.~Flocchini and L.~Moura, editors, \emph{Combinatorial
  Algorithms}, pages 92--106, Cham, 2021. Springer International Publishing.
\newblock ISBN 978-3-030-79987-8.

\bibitem[Bonomo-Braberman and Brito(2023)]{BONOMOBRABERMAN202353}
F.~Bonomo-Braberman and G.~A. Brito.
\newblock Intersection models and forbidden pattern characterizations for
  2-thin and proper 2-thin graphs.
\newblock \emph{Discrete Applied Mathematics}, 339:\penalty0 53--77, 2023.
\newblock ISSN 0166-218X.
\newblock \doi{https://doi.org/10.1016/j.dam.2023.06.013}.
\newblock URL
  \url{https://www.sciencedirect.com/science/article/pii/S0166218X23002354}.

\bibitem[Bonomo-Braberman et~al.(2022)Bonomo-Braberman, Durán, Pardal, and
  Safe]{BONOMOBRABERMAN202243}
F.~Bonomo-Braberman, G.~Durán, N.~Pardal, and M.~D. Safe.
\newblock Forbidden induced subgraph characterization of circle graphs within
  split graphs.
\newblock \emph{Discrete Applied Mathematics}, 323:\penalty0 43--75, 2022.
\newblock ISSN 0166-218X.
\newblock \doi{https://doi.org/10.1016/j.dam.2020.12.021}.
\newblock URL
  \url{https://www.sciencedirect.com/science/article/pii/S0166218X20305473}.
\newblock LAGOS’19 – X Latin and American Algorithms, Graphs, and
  Optimization Symposium – Belo Horizonte, Minas Gerais, Brazil.

\bibitem[Boros et~al.(2004)Boros, Elbassioni, Gurvich, and
  Khachiyan]{BEKEVG2004}
E.~Boros, K.~Elbassioni, V.~Gurvich, and L.~Khachiyan.
\newblock Generating maximal independent sets for hypergraphs with bounded
  edge-intersections.
\newblock In \emph{Latin American Symposium on Theoretical Informatics}, pages
  488--498. Springer, 2004.

\bibitem[Chakraborty and Francis(2020)]{chakraborty2020stab}
D.~Chakraborty and M.~C. Francis.
\newblock On the stab number of rectangle intersection graphs.
\newblock \emph{Theory of Computing Systems}, 64\penalty0 (5):\penalty0
  681--734, 2020.

\bibitem[Chakraborty et~al.(2021)Chakraborty, Das, Francis, and
  Sen]{CHAKRABORTY2021354}
D.~Chakraborty, S.~Das, M.~C. Francis, and S.~Sen.
\newblock On rectangle intersection graphs with stab number at most two.
\newblock \emph{Discrete Applied Mathematics}, 289:\penalty0 354--365, 2021.
\newblock ISSN 0166-218X.
\newblock \doi{https://doi.org/10.1016/j.dam.2020.11.003}.
\newblock URL
  \url{https://www.sciencedirect.com/science/article/pii/S0166218X20304911}.

\bibitem[Chalopin and Gon{\c{c}}alves(2009)]{chalopin2009planar}
J.~Chalopin and D.~Gon{\c{c}}alves.
\newblock Every planar graph is the intersection graph of segments in the
  plane.
\newblock In \emph{Proceedings of the forty-first annual ACM Symposium on
  Theory of Computing}, pages 631--638, 2009.

\bibitem[Cheilaris and Smorodinsky(2012)]{CS2012}
P.~Cheilaris and S.~Smorodinsky.
\newblock Conflict-free coloring with respect to a subset of intervals.
\newblock \emph{arXiv preprint arXiv:1204.6422}, 2012.

\bibitem[Cheilaris et~al.(2012)Cheilaris, Keszegh, and
  P{\'a}lv{\"o}lgyi]{Cheilaris2012}
P.~Cheilaris, B.~Keszegh, and D.~P{\'a}lv{\"o}lgyi.
\newblock \emph{Unique-Maximum and Conflict-Free Coloring for Hypergraphs and
  Tree Graphs}, pages 190--201.
\newblock Springer Berlin Heidelberg, Berlin, Heidelberg, 2012.
\newblock ISBN 978-3-642-27660-6.
\newblock \doi{10.1007/978-3-642-27660-6_16}.
\newblock URL \url{http://dx.doi.org/10.1007/978-3-642-27660-6_16}.

\bibitem[Cheilaris et~al.(2014)Cheilaris, Gargano, Rescigno, and
  Smorodinsky]{CPLGARSS2014}
P.~Cheilaris, L.~Gargano, A.~A. Rescigno, and S.~Smorodinsky.
\newblock Strong conflict-free coloring for intervals.
\newblock \emph{Algorithmica}, 70\penalty0 (4):\penalty0 732--749, Dec. 2014.
\newblock ISSN 0178-4617.
\newblock \doi{10.1007/s00453-014-9929-x}.
\newblock URL \url{http://dx.doi.org/10.1007/s00453-014-9929-x}.

\bibitem[Chen et~al.(2006)Chen, Fiat, Kaplan, Levy, Matou{\v{s}}ek, Mossel,
  Pach, Sharir, Smorodinsky, Wagner, et~al.]{CFKLMMPSSWW2006}
K.~Chen, A.~Fiat, H.~Kaplan, M.~Levy, J.~Matou{\v{s}}ek, E.~Mossel, J.~Pach,
  M.~Sharir, S.~Smorodinsky, U.~Wagner, et~al.
\newblock Online conflict-free coloring for intervals.
\newblock \emph{SIAM Journal on Computing}, 36\penalty0 (5):\penalty0
  1342--1359, 2006.

\bibitem[Chudnovsky et~al.(2003)Chudnovsky, Robertson, Seymour, and
  Thomas]{sgpt02}
M.~Chudnovsky, N.~Robertson, P.~Seymour, and R.~Thomas.
\newblock The strong perfect graph theorem.
\newblock \emph{Annals of mathematics, ISSN 0003-486X, Vol. 164, Nº 1, 2006,
  pags. 51-229}, 164, 01 2003.
\newblock \doi{10.4007/annals.2006.164.51}.

\bibitem[Conforti et~al.(2007)Conforti, Summa, and
  Zambelli]{SetPartitioning2007}
M.~Conforti, M.~D. Summa, and G.~Zambelli.
\newblock Minimally infeasible set-partitioning problems with balanced
  constraints.
\newblock \emph{Mathematics of Operations Research}, 32\penalty0 (3):\penalty0
  497--507, 2007.
\newblock \doi{10.1287/moor.1070.0250}.
\newblock URL \url{http://dx.doi.org/10.1287/moor.1070.0250}.

\bibitem[Corneil et~al.(1995)Corneil, Kim, Natarajan, Olariu, and
  Sprague]{CORNEIL199599}
D.~G. Corneil, H.~Kim, S.~Natarajan, S.~Olariu, and A.~P. Sprague.
\newblock Simple linear time recognition of unit interval graphs.
\newblock \emph{Information Processing Letters}, 55\penalty0 (2):\penalty0 99
  -- 104, 1995.
\newblock ISSN 0020-0190.
\newblock \doi{http://dx.doi.org/10.1016/0020-0190(95)00046-F}.
\newblock URL
  \url{http://www.sciencedirect.com/science/article/pii/002001909500046F}.

\bibitem[Dahllöf et~al.(2004)Dahllöf, Jonsson, and Beigel]{DVPJRB2004}
V.~Dahllöf, P.~Jonsson, and R.~Beigel.
\newblock Algorithms for four variants of the exact satisfiability problem.
\newblock \emph{Theoretical Computer Science}, 320\penalty0 (2):\penalty0 373
  -- 394, 2004.
\newblock ISSN 0304-3975.
\newblock \doi{http://dx.doi.org/10.1016/j.tcs.2004.02.035}.
\newblock URL
  \url{http://www.sciencedirect.com/science/article/pii/S0304397504001343}.

\bibitem[de~Figueiredo et~al.(1995)de~Figueiredo, Meidanis, and
  de~Mello]{DEFIGUEIRE1995}
C.~M. de~Figueiredo, J.~Meidanis, and C.~P. de~Mello.
\newblock A linear-time algorithm for proper interval graph recognition.
\newblock \emph{Information Processing Letters}, 56\penalty0 (3):\penalty0 179
  -- 184, 1995.
\newblock ISSN 0020-0190.
\newblock \doi{http://dx.doi.org/10.1016/0020-0190(95)00133-W}.
\newblock URL
  \url{http://www.sciencedirect.com/science/article/pii/002001909500133W}.

\bibitem[Deng et~al.(1996)Deng, Hell, and Huang]{Deng1996}
X.~Deng, P.~Hell, and J.~Huang.
\newblock Linear-time representation algorithms for proper circular-arc graphs
  and proper interval graphs.
\newblock \emph{SIAM Journal on Computing}, 25\penalty0 (2):\penalty0 390--403,
  1996.

\bibitem[Dhannya and Narayanaswamy(2020)]{dhannya-NSN}
S.~M. Dhannya and N.~S. Narayanaswamy.
\newblock {Perfect Resolution of Conflict-Free Colouring of Interval
  Hypergraphs}.
\newblock In \emph{STACS}, volume 154, pages 52:1--52:16, 2020.
\newblock ISBN 978-3-95977-140-5.
\newblock \doi{10.4230/LIPIcs.STACS.2020.52}.
\newblock URL \url{https://drops.dagstuhl.de/opus/volltexte/2020/11913}.

\bibitem[Dom et~al.(2006)Dom, Guo, Niedermeier, and Wernicke]{Dom2006}
M.~Dom, J.~Guo, R.~Niedermeier, and S.~Wernicke.
\newblock \emph{Minimum Membership Set Covering and the Consecutive Ones
  Property}, pages 339--350.
\newblock Springer Berlin Heidelberg, Berlin, Heidelberg, 2006.
\newblock ISBN 978-3-540-35755-1.
\newblock \doi{10.1007/11785293_32}.
\newblock URL \url{http://dx.doi.org/10.1007/11785293_32}.

\bibitem[Downey and Fellows(2013)]{Downey2013}
R.~G. Downey and M.~R. Fellows.
\newblock \emph{Fundamentals of Parameterized Complexity}.
\newblock Springer Publishing Company, Incorporated, 2013.
\newblock ISBN 1447155580, 9781447155584.

\bibitem[Drori and Peleg(2002)]{DRORI2002}
L.~Drori and D.~Peleg.
\newblock Faster exact solutions for some np-hard problems.
\newblock \emph{Theoretical Computer Science}, 287\penalty0 (2):\penalty0 473
  -- 499, 2002.
\newblock ISSN 0304-3975.
\newblock \doi{http://dx.doi.org/10.1016/S0304-3975(01)00257-2}.
\newblock URL
  \url{http://www.sciencedirect.com/science/article/pii/S0304397501002572}.

\bibitem[Eiter et~al.(2003)Eiter, Gottlob, and Makino]{ETGGKM2003}
T.~Eiter, G.~Gottlob, and K.~Makino.
\newblock New results on monotone dualization and generating hypergraph
  transversals.
\newblock \emph{SIAM Journal on Computing}, 32\penalty0 (2):\penalty0 514--537,
  2003.

\bibitem[Elbassioni and Rauf(2010)]{EKIR2010}
K.~Elbassioni and I.~Rauf.
\newblock Polynomial-time dualization of r-exact hypergraphs with applications
  in geometry.
\newblock \emph{Discrete Mathematics}, 310\penalty0 (17):\penalty0 2356--2363,
  2010.

\bibitem[Erd\H{o}s(1963)]{erdos1963}
P.~Erd\H{o}s.
\newblock On a combinatorial problem.
\newblock \emph{Nordisk Matematisk Tidskrift}, 11\penalty0 (1):\penalty0 5--10,
  1963.
\newblock ISSN 00291412.
\newblock URL \url{http://www.jstor.org/stable/24524364}.

\bibitem[Erd\H{o}s et~al.(1961)Erd\H{o}s, Ko, and Rado]{erdos1961intersection}
P.~Erd\H{o}s, C.~Ko, and R.~Rado.
\newblock Intersection theorems for systems of finite sets.
\newblock \emph{The Quarterly Journal of Mathematics}, 12:\penalty0 313--320,
  1961.

\bibitem[Erd\H{o}s et~al.(1966)Erd\H{o}s, Goodman, and
  P{\'o}sa]{erdos1966representation}
P.~Erd\H{o}s, A.~W. Goodman, and L.~P{\'o}sa.
\newblock The representation of a graph by set intersections.
\newblock \emph{Canadian Journal of Mathematics}, 18\penalty0
  (106-112):\penalty0 86, 1966.

\bibitem[Even et~al.(2002)Even, Lotker, Ron, and Smorodinsky]{FOCS2002}
G.~Even, Z.~Lotker, D.~Ron, and S.~Smorodinsky.
\newblock Conflict-free colorings of simple geometric regions with applications
  to frequency assignment in cellular networks.
\newblock \emph{2013 IEEE 54th Annual Symposium on Foundations of Computer
  Science}, 0:\penalty0 691, 2002.
\newblock ISSN 0272-5428.
\newblock \doi{http://doi.ieeecomputersociety.org/10.1109/SFCS.2002.1181994}.

\bibitem[Even et~al.(2003)Even, Lotker, Ron, and Smorodinsky]{ELRS2003}
G.~Even, Z.~Lotker, D.~Ron, and S.~Smorodinsky.
\newblock Conflict-free colorings of simple geometric regions with applications
  to frequency assignment in cellular networks.
\newblock \emph{SIAM Journal on Computing}, 33\penalty0 (1):\penalty0 94--136,
  2003.

\bibitem[Fernandez de~la Vega et~al.(1992)Fernandez de~la Vega, Paschos, and
  Saad]{Fern1992}
W.~Fernandez de~la Vega, V.~T. Paschos, and R.~Saad.
\newblock \emph{Average case analysis of a greedy algorithm for the minimum
  hitting set problem}, pages 130--138.
\newblock Springer Berlin Heidelberg, Berlin, Heidelberg, 1992.
\newblock ISBN 978-3-540-47012-0.
\newblock \doi{10.1007/BFb0023824}.
\newblock URL \url{http://dx.doi.org/10.1007/BFb0023824}.

\bibitem[Gardi(2007)]{Gardi2007}
F.~Gardi.
\newblock The roberts characterization of proper and unit interval graphs.
\newblock \emph{Discrete Mathematics}, 307\penalty0 (22):\penalty0 2906 --
  2908, 2007.
\newblock ISSN 0012-365X.
\newblock \doi{http://doi.org/10.1016/j.disc.2006.04.043}.
\newblock URL
  \url{http://www.sciencedirect.com/science/article/pii/S0012365X07000696}.

\bibitem[Gardi(2011)]{FGardi}
F.~Gardi.
\newblock On partitioning interval graphs into proper interval subgraphs and
  related problems.
\newblock \emph{Journal of Graph Theory}, 68\penalty0 (1):\penalty0 38--54,
  2011.
\newblock ISSN 1097-0118.
\newblock \doi{10.1002/jgt.20539}.
\newblock URL \url{http://dx.doi.org/10.1002/jgt.20539}.

\bibitem[Garey and Johnson(1979)]{Garey1979}
M.~R. Garey and D.~S. Johnson.
\newblock \emph{Computers and Intractability: A Guide to the Theory of
  NP-Completeness}.
\newblock W. H. Freeman \& Co., New York, NY, USA, 1979.
\newblock ISBN 0716710447.

\bibitem[Gargano and Rescigno(2015)]{GLAR2015}
L.~Gargano and A.~Rescigno.
\newblock Complexity of conflict-free colorings of graphs.
\newblock \emph{Theoretical Computer Science}, 566:\penalty0 39 -- 49, 2015.
\newblock ISSN 0304-3975.
\newblock \doi{http://dx.doi.org/10.1016/j.tcs.2014.11.029}.

\bibitem[Gavril(1974)]{gavril1974intersection}
F.~Gavril.
\newblock The intersection graphs of subtrees in trees are exactly the chordal
  graphs.
\newblock \emph{Journal of Combinatorial Theory, Series B}, 16\penalty0
  (1):\penalty0 47--56, 1974.

\bibitem[Gavril(1978)]{gavril1978recognition}
F.~Gavril.
\newblock A recognition algorithm for the intersection graphs of paths in
  trees.
\newblock \emph{Discrete Mathematics}, 23\penalty0 (3):\penalty0 211--227,
  1978.

\bibitem[Gilmore and Hoffman(2011)]{GilHof2011}
P.~C. Gilmore and A.~J. Hoffman.
\newblock \emph{A Characterization of Comparability graphs and of Interval
  graphs}, pages 65--74.
\newblock World Scientific, 2011.
\newblock \doi{10.1142/9789812796936_0006}.
\newblock URL
  \url{http://www.worldscientific.com/doi/abs/10.1142/9789812796936_0006}.

\bibitem[Golumbic(1985)]{Gol1985}
M.~C. Golumbic.
\newblock Interval graphs and related topics.
\newblock \emph{Discrete Mathematics}, 55\penalty0 (2):\penalty0 113 -- 121,
  1985.
\newblock ISSN 0012-365X.
\newblock \doi{https://doi.org/10.1016/0012-365X(85)90039-1}.
\newblock URL
  \url{http://www.sciencedirect.com/science/article/pii/0012365X85900391}.

\bibitem[Golumbic(2004)]{Gol2004}
M.~C. Golumbic.
\newblock \emph{Algorithmic Graph Theory and Perfect Graphs (Annals of Discrete
  Mathematics, Vol 57)}.
\newblock North-Holland Publishing Co., Amsterdam, The Netherlands, 2004.
\newblock ISBN 0444515305.

\bibitem[Golumbic et~al.(2009)Golumbic, Lipshteyn, and
  Stern]{Gol2009weakchordal}
M.~C. Golumbic, M.~Lipshteyn, and M.~Stern.
\newblock Intersection models of weakly chordal graphs.
\newblock \emph{Discrete Applied Mathematics}, 157\penalty0 (9):\penalty0 2031
  -- 2047, 2009.
\newblock ISSN 0166-218X.
\newblock \doi{http://doi.org/10.1016/j.dam.2008.11.017}.
\newblock URL
  \url{http://www.sciencedirect.com/science/article/pii/S0166218X08005349}.
\newblock Optimal Discrete Structures and AlgorithmsODSA 2006.

\bibitem[Gr{\"o}tschel et~al.(2012)Gr{\"o}tschel, Lov{\'a}sz, and
  Schrijver]{grotschel2012}
M.~Gr{\"o}tschel, L.~Lov{\'a}sz, and A.~Schrijver.
\newblock \emph{Geometric algorithms and combinatorial optimization}, volume~2.
\newblock Springer Science \& Business Media, 2012.

\bibitem[Grötschel et~al.(1984)Grötschel, Lovász, and
  Schrijver]{GROTSCHEL1984325}
M.~Grötschel, L.~Lovász, and A.~Schrijver.
\newblock Polynomial algorithms for perfect graphs.
\newblock In C.~Berge and V.~Chvátal, editors, \emph{Topics on Perfect
  Graphs}, volume~88 of \emph{North-Holland Mathematics Studies}, pages 325 --
  356. North-Holland, 1984.
\newblock \doi{https://doi.org/10.1016/S0304-0208(08)72943-8}.
\newblock URL
  \url{http://www.sciencedirect.com/science/article/pii/S0304020808729438}.

\bibitem[Har-Peled and Smorodinsky(2005)]{HSSS2005}
S.~Har-Peled and S.~Smorodinsky.
\newblock Conflict-free coloring of points and simple regions in the plane.
\newblock \emph{Discrete \& Computational Geometry}, 34\penalty0 (1):\penalty0
  47--70, 2005.

\bibitem[Harary(1969)]{Harary}
F.~Harary.
\newblock \emph{Graph Theory}.
\newblock Addison-Wesley Series in Mathematics. Addison Wesley, 1969.

\bibitem[Horev et~al.(2010)Horev, Krakovski, and Smorodinsky]{HERKSS2010}
E.~Horev, R.~Krakovski, and S.~Smorodinsky.
\newblock Conflict-free coloring made stronger.
\newblock In \emph{Proceedings of the 12th Scandinavian Conference on Algorithm
  Theory}, SWAT'10, pages 105--117, Berlin, Heidelberg, 2010. Springer-Verlag.
\newblock ISBN 3-642-13730-X, 978-3-642-13730-3.
\newblock \doi{10.1007/978-3-642-13731-0_11}.
\newblock URL \url{http://dx.doi.org/10.1007/978-3-642-13731-0_11}.

\bibitem[Karp(1972)]{Karp1972}
R.~M. Karp.
\newblock \emph{Reducibility among Combinatorial Problems}, pages 85--103.
\newblock Springer US, Boston, MA, 1972.
\newblock ISBN 978-1-4684-2001-2.
\newblock \doi{10.1007/978-1-4684-2001-2_9}.
\newblock URL \url{http://dx.doi.org/10.1007/978-1-4684-2001-2_9}.

\bibitem[Katz et~al.(2012)Katz, Lev-Tov, and Morgenstern]{KATZ}
M.~J. Katz, N.~Lev-Tov, and G.~Morgenstern.
\newblock Conflict-free coloring of points on a line with respect to a set of
  intervals.
\newblock \emph{Computational Geometry}, 45\penalty0 (9):\penalty0 508--514,
  2012.

\bibitem[Keller and Smorodinsky(2020)]{Keller2018}
C.~Keller and S.~Smorodinsky.
\newblock Conflict-free coloring of intersection graphs of geometric objects.
\newblock \emph{Discrete \& Computational Geometry}, 64\penalty0 (3):\penalty0
  916--941, 2020.

\bibitem[Kuhn et~al.(2005)Kuhn, von Rickenbach, Wattenhofer, Welzl, and
  Zollinger]{Kuhn2005}
F.~Kuhn, P.~von Rickenbach, R.~Wattenhofer, E.~Welzl, and A.~Zollinger.
\newblock Interference in cellular networks: The minimum membership set cover
  problem.
\newblock In \emph{Proceedings of the 11th Annual International Conference on
  Computing and Combinatorics}, pages 188--198, Berlin, Heidelberg, 2005.
  Springer-Verlag.
\newblock ISBN 3-540-28061-8, 978-3-540-28061-3.
\newblock \doi{10.1007/11533719_21}.
\newblock URL \url{http://dx.doi.org/10.1007/11533719_21}.

\bibitem[Lekkerkerker and Boland(1962)]{lekkeikerker1962}
C.~Lekkerkerker and J.~Boland.
\newblock Representation of a finite graph by a set of intervals on the real
  line.
\newblock \emph{Fundamenta Mathematicae}, 51\penalty0 (1):\penalty0 45--64,
  1962.

\bibitem[Lev-Tov and Peleg(2009)]{NLDP2009}
N.~Lev-Tov and D.~Peleg.
\newblock Conflict-free coloring of unit disks.
\newblock \emph{Discrete Appl. Math.}, 157\penalty0 (7):\penalty0 1521--1532,
  Apr. 2009.
\newblock ISSN 0166-218X.
\newblock \doi{10.1016/j.dam.2008.09.005}.
\newblock URL \url{http://dx.doi.org/10.1016/j.dam.2008.09.005}.

\bibitem[L{\'e}v{\^e}que et~al.(2009)L{\'e}v{\^e}que, Maffray, and
  Preissmann]{leveque2009}
B.~L{\'e}v{\^e}que, F.~Maffray, and M.~Preissmann.
\newblock Characterizing path graphs by forbidden induced subgraphs.
\newblock \emph{Journal of Graph Theory}, 62\penalty0 (4):\penalty0 369--384,
  2009.

\bibitem[Lévêque et~al.(2009)Lévêque, Maffray, and Preissmann]{JGT20407}
B.~Lévêque, F.~Maffray, and M.~Preissmann.
\newblock Characterizing path graphs by forbidden induced subgraphs.
\newblock \emph{Journal of Graph Theory}, 62\penalty0 (4):\penalty0 369--384,
  2009.
\newblock ISSN 1097-0118.
\newblock \doi{10.1002/jgt.20407}.
\newblock URL \url{http://dx.doi.org/10.1002/jgt.20407}.

\bibitem[McKee and McMorris(1999)]{mckee1999}
T.~A. McKee and F.~R. McMorris.
\newblock \emph{Topics in intersection graph theory}.
\newblock SIAM, 1999.

\bibitem[Mertzios et~al.(2010)Mertzios, Sau, and Zaks]{Mertzios2010}
G.~Mertzios, I.~Sau, and S.~Zaks.
\newblock The recognition of tolerance and bounded tolerance graphs is
  np-complete.
\newblock \emph{SIAM Journal on Computing}, 40, 01 2010.
\newblock \doi{10.1137/090780328}.

\bibitem[Monma and Wei(1986)]{monma1986}
C.~L. Monma and V.~K. Wei.
\newblock Intersection graphs of paths in a tree.
\newblock \emph{Journal of Combinatorial Theory, Series B}, 41\penalty0
  (2):\penalty0 141--181, 1986.

\bibitem[Mustafa and Ray(2010)]{mustafa2010}
N.~H. Mustafa and S.~Ray.
\newblock Improved results on geometric hitting set problems.
\newblock \emph{Discrete \& Computational Geometry}, 44\penalty0 (4):\penalty0
  883--895, 2010.

\bibitem[Pach and Tardos(2008)]{PJGT2008}
J.~Pach and G.~Tardos.
\newblock Computational geometry and graph theory.
\newblock In H.~Ito, M.~Kano, N.~Katoh, and Y.~Uno, editors,
  \emph{"Computational Geometry and Graph Theory"}, chapter Coloring
  Axis-Parallel Rectangles, pages 178--185. Springer-Verlag, Berlin,
  Heidelberg, 2008.
\newblock ISBN 978-3-540-89549-7.
\newblock \doi{10.1007/978-3-540-89550-3_19}.
\newblock URL \url{http://dx.doi.org/10.1007/978-3-540-89550-3_19}.

\bibitem[Pach and Tardos(2009)]{PJGT2009}
J.~Pach and G.~Tardos.
\newblock Conflict-free colourings of graphs and hypergraphs.
\newblock \emph{Combinatorics, Probability and Computing}, 18\penalty0
  (05):\penalty0 819--834, 2009.

\bibitem[Roberts(1978)]{roberts1978graph}
F.~S. Roberts.
\newblock \emph{Graph theory and its applications to problems of society},
  pages 27--31.
\newblock SIAM, 1978.

\bibitem[Schrijver(1986)]{Schrijver}
A.~Schrijver.
\newblock \emph{Theory of Linear and Integer Programming}.
\newblock John Wiley \& Sons, Inc., New York, NY, USA, 1986.
\newblock ISBN 0-471-90854-1.

\bibitem[Schrijver(2003)]{schrijver-book}
A.~Schrijver.
\newblock \emph{Combinatorial Optimization - Polyhedra and Efficiency}.
\newblock Springer, 2003.

\bibitem[Smorodinsky(2003)]{Sm2003PhD}
S.~Smorodinsky.
\newblock \emph{Combinatorial problems in computational geometry}.
\newblock PhD thesis, Tel-Aviv University, 2003.

\bibitem[Smorodinsky(2007)]{Sm2007}
S.~Smorodinsky.
\newblock On the chromatic number of geometric hypergraphs.
\newblock \emph{SIAM Journal on Discrete Mathematics}, 21\penalty0
  (3):\penalty0 676--687, 2007.

\bibitem[Smorodinsky(2013)]{Sm2013}
S.~Smorodinsky.
\newblock Conflict-free coloring and its applications.
\newblock In \emph{Geometry—Intuitive, Discrete, and Convex}, pages 331--389.
  Springer, 2013.

\bibitem[Vandal et~al.(2009)Vandal, Conder, and Gentleman]{Vandal2009}
A.~Vandal, M.~Conder, and R.~Gentleman.
\newblock Minimal covers of maximal cliques for interval graphs.
\newblock \emph{Ars Combinatoria}, 92, 07 2009.

\bibitem[West(2000)]{West}
D.~B. West.
\newblock \emph{Introduction to Graph Theory}.
\newblock Prentice Hall, 2 edition, September 2000.
\newblock ISBN 0130144002.

\end{thebibliography}

\end{document}